%% file: arxiv.tex
\documentclass[a4paper,11pt]{article}
\usepackage{booktabs}
\usepackage{amsmath,amssymb,graphicx,xcolor,pgfplots,enumitem,comment,csquotes,subcaption}
\pgfplotsset{compat=1.17} 
\usepackage{amsthm}
\usepackage[margin=1in]{geometry}
\usepackage{authblk}
\usepackage{hyperref,cleveref}

\renewcommand{\l}{\left(}
\renewcommand{\r}{\right)}

\newcommand{\ol}{\overline}

\newtheorem*{theorem*}{Theorem}
\Crefname{notat}{Notation}{Notations}

\newtheorem{theorem}{Theorem}[section]
\newtheorem{lemma}[theorem]{Lemma}

\theoremstyle{definition}
\newtheorem{definition}[theorem]{Definition}
\theoremstyle{remark}
\newtheorem{remark}{Remark}
\newtheorem{example}[theorem]{Example}

\title{Dynamical Analysis of the EIP-1559 Ethereum Fee Market}
\author[1]{Stefanos Leonardos\thanks{In alphabetical order.}}

\newcommand\CoAuthorMark{\footnotemark[\arabic{footnote}]}
\author[2]{Barnab\'e Monnot\protect\CoAuthorMark}
\author[1]{Daniel Reijsbergen\protect\CoAuthorMark}
\author[1]{Stratis Skoulakis\protect\CoAuthorMark}
\author[1]{Georgios Piliouras}
\affil[1]{Singapore University of Technology and Design, \emph{\{stefanos\_leonardos,daniel\_reijsbergen,efstratios,georgios\}@sutd.edu.sg}}
\affil[2]{Ethereum Foundation, \emph{barnabe.monnot@ethereum.org}}
\date{}

\graphicspath{{Figures/}{./}{Simulations/Matlab_files/}{Simulations/Python_notebooks/}}

\begin{document}


\maketitle

\begin{abstract}

Participation in permissionless blockchains results in competition over system resources, which  needs to be controlled with fees.
Ethereum's current fee mechanism is implemented via a first-price auction that results in unpredictable fees as well as other inefficiencies.
EIP-1559 is a recent, improved proposal that introduces a number of innovative features such as a dynamically adaptive basefee that is burnt, instead of being paid to the miners. Despite intense interest in understanding its properties, 
several basic questions such as whether and under what conditions does this protocol self-stabilize  have remained elusive thus far.  \par
We perform a thorough analysis of the resulting fee market dynamic mechanism 
 via a combination of tools from game theory as well as dynamical systems. We start by providing bounds on the step-size of the base-fee update rule that suffice for global convergence to equilibrium via Lyapunov arguments. In the negative direction, we show that for larger step-sizes instability and even formal Li-Yorke chaos are possible under a wide range of settings. We complement these topological results with quantitative bounds on the possible range of basefees.  We conclude our analysis with a thorough experimental case study that corroborates our theoretical findings.
\end{abstract}

\input{introduction}
\input{model}
\input{theoretical_analysis}
\input{experiments}
\input{conclusions}

\section*{Acknowledgements}

Stefanos Leonardos and Stratis Skoulakis acknowledge support from NRF 2018 Fellowship NRF-NRFF2018-07. Barnab\'e Monnot acknowledges support from the Ethereum Foundation. Daniel Reijsbergen acknowledges NRF Award No.\ NSoE DeST-SCI2019-0009.
Georgios Piliouras acknowledges support from NRF2019-NRF- ANR095 ALIAS grant, grant PIE-SGP-AI-2018-01, NRF 2018 Fellowship NRF-NRFF2018-07 and the Ethereum Foundation.

\bibliographystyle{plain}
\bibliography{bib_auctions}

\appendix

\end{document}

%% file: introduction.tex
\section{Introduction}\label{sec:intro}


The emergence of decentralized, Turing-complete blockchains, such as Ethereum~\cite{wood2014ethereum}, ushered in the possibility of creating alternative economic systems, where traditional institutions (such as exchanges, banks, e.t.c.) are implemented in open-source code and where the state of the system/universal computer is stored in an immutable public blockchain. The extreme versatility of such systems, at least in terms of their fundamental capabilities, naturally raises a lot of critical design considerations as these abstract ideas are fleshed out into concrete implementations. Moreover, as participation in these systems steadily increases over time, these initial designs face novel demands and some careful adaptation becomes necessary.





Arguably, one of the most critical 
 real-world design decisions in Ethereum, as well as in any other programmable blockchain, is how the protocol decides on the costs/rewards structure for the different types of  
 participating entities. The protocol charges users fees for having their transactions processed by the network and included in the blockchain. These transactions fees are typically referred to as ``gas fees".
 These fees are then distributed to the miners rewarding them for dedicating  computational resources to preserving the safety of the blockchain.
 
 Ethereum's current fee system has been recognized as an important design challenge. The issue primarily lies on the decision to set fees by using a simple first price auction mechanism. Effectively, all users submit their bids in regards to how much they are willing to pay to have their transactions included in the blockchain and the miners 
 typically select the highest priced entries for inclusion given the block capacity constraints. Due to the non-truthful nature of first price auctions, choosing an appropriate bidding fee is a non-trivial task and  users can end up significantly overpaying for system participation. 
 
 From a traditional mechanism design perspective, the solution to the aforementioned problem seems relatively straightforward: Replace the first price auction with either a Vickrey – Clarke – Groves (VCG) auction~\cite{Nis07,Var07} or a (generalized) second price auction~\cite{Bre12,Car15}, which reduce the strategic complexity on the side of the bidders, lead to more efficient outcomes and are known to work well in practice (e.g., internet advertising). Unfortunately, such approaches can be easily exploited and gamed by miners who can artificially increase demands for their blocks, increasing the resulting fees while decreasing meaningful system participation. Moreover, such mechanisms are vulnerable to collusion~\cite{conitzer2006failures,bachrach2010honor}.

Recently, a new proposal (EIP-1559) has been put forward to address these issues~\cite{Con19}.
 A key aspect of this mechanism is the introduction of a \textit{basefee that is automatically adjusted by the protocol} depending on how congested the network is. This basefee effectively
 plays the role of a reserve price, matching supply and
demand.
 Critically, this basefee is burned, which prevents the emergence of perverse incentives where miners can extract increased fees from the users by acting dishonestly. Users 
  seeking fast inclusion of their transactions
can supplement the basefee with a tip, which is the only fee that is received by the miners. An economic analysis of EIP-1559 has identified desirable properties, e.g., it is incentive compatible for myopic miners and 
   as well as for users \textit{except during time periods with excessively low base fees}~\cite{roughgarden2020transaction}. Of course, to provide insights about whether such conditions will be satisfied in practice an economic analysis alone is not sufficient as one needs to explicitly analyze the dynamic evolution of the mechanism parameters over time. This raises our driving question: \textit{Under which conditions do the EIP-1559 dynamics self-stabilize? When these conditions are not satisfied how complex, unpredictable can the resulting behavior be?} 
 

{\bf Our results.}
We perform both a theoretical as well as experimental analysis of the dynamics and  stability properties of the EIP-1559 protocol. In particular, we investigate not only sufficient conditions for network stability and convergence to equilibrium but furthermore, we provide for the first-time to our knowledge 
a stress-test type of analysis where we push the system parameters past its stable range and prove phase transitions/bifurcations as well as the formal onset of chaos. \par
Our main observation is that the basefee adjustment parameter (step-size) has a critical impact in the stability of the system. In the theoretical part of the paper (\Cref{sec:analysis}), we provide threshold bounds for the step-size which allow the system to stabilize (\Cref{t:threshold_convergence}). For larger values of the step-size (or of the other critical parameters of the system, transaction demand and user valuations), we show that the basefee dynamics may become formally chaotic (\Cref{thm:chaos}). However, even in this unstable regime, the basefee remains within a bounded region and is relatively well behaved. By contrast, adverse effects are observed in the block occupancies (which may oscillate between their extremes, full to empty and vise versa).\par
On the experimental side we validate our theoretical findings by showcasing high variance periods where blocks alternate between full and empty state and basefee spikes up and down, using a fee market simulation library with agent-based components (\Cref{sec:experiments}). We first look at the impact of three variables on the prevalence of these high variance periods: the demand variance, or how ``noisy'' the demand process is; the initial condition of the demand process, from just enough to fill blocks entirely to twice that demand; and the tolerance of the transaction pool eviction policy, with more tolerant pools keeping transactions even as their fee cap stands at a lower value than the basefee. We find all three variables positively correlate with more appearances of high variance periods, highlighting the forces inducing variance in the fee market. We additionally find that using stricter pool eviction policies hurts user efficiency and miner revenue, casting doubt on the incentive compatibility of this strategy to yield more stable basefee updates.

%% file: model.tex
\section{Model}\label{sec:model}
Our description of the model consists of two parts. The first describes the transaction fee mechanism of EIP1559 with a special focus on the dynamic adjustment of the basefee (\Cref{sub:eip}) and the second describes the assumptions concerning agents' behavior that we study in this paper (\Cref{sub:users,sub:fluid}). 

\subsection{Transaction Fee Market and EIP1559}\label{sub:eip}
We consider a blockchain-enabled economy in which users make transactions over a distributed network.\footnote{Our analysis is based on the Ethereum blockchain. However, there are other blockchains, such as Filecoin \cite{Fil20}, that implement very similar mechanisms and the main ideas of our results (up to technical details) readily extend to these settings as well.} Users submit their transactions to a common pool together with a bid which specifies how much they are willing to pay for the computational resources that are required for their transactions to be processed. The transactions, along with the bids, are viewed by the miners who select which transactions to include in the blocks that they create. In existing mechanisms (including Ethereum's current economic model), bids comprise a single transaction fee. Miners can sort the transactions and typically select the ones with the highest fees. The miner who will include a transaction in a valid block receives the entire fee in a process that closely resembles a generalized first price auction. \par
According to the proposed reform of EIP1559, bids comprise two elements $(f,p)$: (i) the \emph{feecap}, $f$, which is the maximum amount that the user is willing to pay for their transaction to be processed, and (ii) the \emph{premium}, $p$, which is the maximum tip that the user is willing to pay to the miner who will eventually process their transaction. In particular, a user who will get their transaction included in the blockchain will never pay more than the feecap and the miner who will process the transaction will never receive more than the premium. \par
The main element of EIP1559 and its main difference from existing mechanisms is the stipulation of a dynamically adjusted \emph{basefee}, $b_t, t>0$, where $t$ denotes the block height. Every transaction that gets included in a block $B_t, t>0$ needs to pay the basefee, $b_t$, that is valid at that block. Instead of being transferred from the user to the miner, the basefee is \emph{burnt}, i.e., it is permanently removed from the circulating supply of the native currency (e.g., ETH). For each included transaction, miners will receive the minimum between the premium and the difference between the feecap and the basefee. Specifically, the \emph{miner's tip} is defined by
\begin{equation}\label{eq:tip}
\text{miner's tip}: =\min{\{f-b_t,p\}}.
\end{equation}
Blocks have size $T$ and, in EIP1559, a target block load $T/2$.\footnote{Size is measured in gas, i.e., typically $T$ denotes the gas limit. Here, we express all measurements in units per gas, so under the assumption that all transactions use the same amount of gas, one may think of $T$ as number of transactions.} Let $g_t$ denote the number of transactions that get included in block $B_t$. Since $g_t$ depends on the basefee, $b_t$, we will write $g_t\mid b_t$ to denote the transactions that get included in $B_t$ given that the basefee is equal to $b_t$. The basefee is updated after every block according to the following equation
\begin{equation}\label{eq:basefee}
b_{t+1}=b_t \l 1+d\cdot\frac{g_t\mid b_t-T/2}{T/2}\r, \qquad \text{for any } t\in\mathbb N.
\end{equation}
where $d$ denotes an adjustment factor (or step size), currently set at $d=0.125$ \cite{Rob20e,Mon20}. Equation \eqref{eq:basefee} suggests that the basefee will increase if the load of block $B_t$ is larger than the target block load, i.e., if there is increasing demand or congestion in the system, and will decrease otherwise. The magnitude of the change is regulated by the excess (shortage) of transaction load compared to the target load (currently $T/2$) and parameter $d$. Our main goal in this paper is to analyze the stability and properties of the dynamical system that is determined by equation \eqref{eq:basefee}.

\subsection{Behavioral Model: Miners and Users}\label{sub:users}
In general, we will assume that users (transactions) arrive to the pool according to a random process. We will write $N_t$ to denote the random number of transactions that arrive between two consecutive blocks $B_t,B_{t+1}$ for $t\ge0$. We assume that $N_t\sim \mathcal P(\lambda T)$ for any $t\ge0$, where $\mathcal P(\lambda)$ denotes the Poisson distribution of parameter $\lambda T$. To avoid trivial cases, we will assume that $\lambda >1/2$, i.e., that the arrival rate is larger than the target block load. For the theoretical analysis, we will assume that users leave the pool if their transaction is not included in the next block and return according to the specified arrival process.\footnote{This assumption only reduces unnecessary complexities in the analysis and is relaxed in the simulations without significant effect in the results.} Whenever necessary, we will index users (transactions) with $i,j\in\mathbb N$.\par
As mentioned above, miners view all transactions in the pool along with their bids, $(f,p)$, and decide which transactions to include in the blocks that they mine. We assume that miners are willing to process transactions only if the fees that they receive are at least some commonly known $\epsilon>0$. This is due to the intrinsic marginal cost for miners to include the transaction. For instance, each transaction increases the size (in bytes) of the block and its propagation time over the network of miners, leading to an increase in the risk of producing a stale block (called \emph{uncle}).\footnote{It is expected for user wallets to encode this default $\epsilon$ in their fee estimation strategies, thus supporting common knowledge among fee market participants. At the moment, the value of 1 nanoETH ($10^{-9}$ ETH) is recommended as such a default by the EIP itself \cite{But19}.} Thus, miners will select a transaction to be included in block $B_t$ if $f\ge b_t+\epsilon$, and $p\ge \epsilon$, i.e., if the feecap is large enough to cover both the basefee and the minimum acceptable premium for miners, and the premium is large enough to satisfy the miner's tip. These conditions are summarized in the following minimum \emph{inclusion requirement}
\begin{equation}\label{eq:include}
\text{miner's tip}=\min{\{f-b_t,p\}}\ge\epsilon.
\end{equation} 
Finally, each user $i\in \mathbb N$ has a valuation $v_i$ which is drawn from some common (for all users) distribution function $v\sim F$ with strictly positive support $S\subseteq \mathbb R_+$. For convenience, we assume throughout that $F$ is continuous and strictly increasing (i.e., non-atomic). We will write $\ol{F}$ to denote the \emph{survival function} of $F$, i.e., $\ol{F}(x)=1-F(x),$ for any $x\in \mathbb R$. For the most part, we will assume that users are \emph{non-strategic}, which means that they bid their valuations as feecaps, i.e., $f_i=v_i$ and that they set a premium equal to the miner's acceptance threshold $\epsilon$, i.e., $p_i=\epsilon$ for each user $i\in \mathbb N$. In short, a non-strategic user with valuation $v_i$, is defined as a user who bids $(f,p)=(v_i,\epsilon)$. 

\subsection{Non-Atomic Model}\label{sub:fluid}
Based on the above assumptions, the dynamical system that is determined by equation \eqref{eq:basefee} is a \emph{discrete time, discrete space stochastic process} $\{b_t\}_{t\ge 0}$. The source of randomness is the term $g_t\mid b_t$, i.e., the number of transactions that get included in block $B_t$ given the basefee $b_t$. However, for the most part of the analysis, it will be sufficient to consider a non-atomic (or fluid) approximation of the above system. Accordingly, we will assume that there is a fixed number of $\lambda T$ arrivals between each two consecutive blocks and that the fraction of users who are willing to pay the basefee (plus the miners' premium) is equal to $\lambda T\ol{F}(b_t+\epsilon)$. Taking into account that a block has a maximum size, $T$, the above lead to the following \emph{discrete time, continuous space deterministic process}
\begin{equation*}
b_{t+1}=b_t+b_t\frac{d}{T/2}\l\min{\{T,\lambda T\ol{F}(b_t+\epsilon)\}}-T/2\r.
\end{equation*}
which after some straightforward simplifications leads to 
\begin{equation}\label{eq:fluid}
b_{t+1}=b_t+b_td\min{\{1,2\lambda \ol{F}(b_t+\epsilon)-1\}}.
\end{equation}
The analysis of the dynamical system $\{b_t\}_{t\ge0}$ that is defined by equation \eqref{eq:fluid} will be the main subject of the theoretical part of this paper. In the simulations, we again employ the discrete model described above. 

\begin{remark}\label{rem:approximation}
For practical purposes, the approximation of discrete arrivals by a continuous process is justified by the fact that $g_t$ actually denotes \emph{gas units}, which offer a much larger granularity than exact numbers of transactions. Moreover, the number of arrivals between consecutive blocks can be considered fairly constant during stationary periods which are of interest here. If the demand shifts to a new stationary level, then the basefee is also expected to shift to adjust to this new level. From a theoretical perspective, the deterministic dynamical system in equation \eqref{eq:fluid} can be justified as an approximation of the sequence of conditional expectations $\mathbb E[b_{t+1}|b_t]$ as explained in \Cref{lem:preliminary} below.
\end{remark}

%% file: theoretical_analysis.tex
\section{Analysis}\label{sec:analysis}

Our main task in this section is to analyze the convergence and stability properties of the dynamical system $\{b_t\}_{t\ge0}$ of equation \eqref{eq:fluid}.

\subsection{Preliminaries}\label{sub:preliminaries}

As mentioned in \Cref{rem:approximation}, our first observation is that, the deterministic dynamical system in equation \eqref{eq:fluid} can be justified as an approximation of the sequence of conditional expectations $\mathbb E[b_{t+1}|b_t]$. In particular, the base fee of block $t+1$ depends only on the state of block $t$, which means that the \emph{stochastic process} $\{b_t\}, t\ge0$ has the \emph{Markov property}.\footnote{Formally, a stochastic process $X_t, t\ge 0$ is \emph{Markovian}, with respect to a filtration ${\mathcal F}_t=\sigma(X_s\mid s\le t)$, if for any fixed time $t\ge0$, the future of the process, i.e., $X_{t+1}$, is independent of $\mathcal F_t$ given $X_t$.} This allows us to derive a closed form formula for the conditional expectation $\mathbb E[b_{t+1}\mid b_1,b_2,\dots,b_t]$ as shown in \Cref{lem:preliminary}.

\begin{lemma}\label{lem:preliminary}
Suppose that the number $N_t$ of transactions that arrive to the transaction pool between consecutive blocks $B_t,B_{t+1}$ follows a Poisson process with rate $\lambda T$, with $\lambda >1/2$ for any $t\ge0$. Further, suppose that users valuations $v_i, i\in \mathbb N$ are independently drawn from a common distribution $v\sim F$ for some continuous and strictly increasing distribution function $F$ and that users are nonstrategic, i.e., that their bids satisfy $(f,p)=(v_i,\epsilon)$. Then, it holds that the stochastic process $\{b_{t}\}_{t\ge0}$ of equation \eqref{eq:basefee}
\[
b_{t+1}=b_t \l 1+d\cdot\frac{g_t\mid b_t-T/2}{T/2}\r, \qquad \text{for any } t\in\mathbb N.\]
has the Markov property and
\begin{equation}\label{eq:expectation}
\mathbb E[b_{t+1}\mid b_t]\le b_t+b_td\min{\{1, 2\lambda \ol{F}(b_t+\epsilon)-1\}}.
\end{equation}
\end{lemma}

\begin{proof}
The Markov property is immediate from the definition of $b_{t+1}$ since $b_{t+1}$ is fully determined by $b_t$ and $g_t\mid b_t$. Thus, $\mathbb E[b_{t+1}\mid b_0,\dots,b_t]= \mathbb E[b_{t+1}\mid b_t]$ for any $t\ge0$, with
\begin{align*}
    \mathbb E[b_{t+1}\mid b_t]&=\mathbb E\left[b_t \l 1+d\cdot\frac{g_t\mid b_t-T/2}{T/2}\r\Big | b_t\right]\\&=b_t \l 1+d\cdot\frac{\mathbb{E}[g_t\mid b_t]-T/2}{T/2}\r.
\end{align*}
To proceed with the calculation of the conditional expectations $\mathbb{E}[g_t\mid b_t]$, we define the random variables
\[X_i=\begin{cases}1, & \text{if $i$'s valuation satisfies the inclusion requirement}\\ 
&\text{in block $B_t$},\\0,& \text{otherwise.}\end{cases}\]
Recall from equation \eqref{eq:include}, that the minimum inclusion requirement is that $\min{\{f-b_t,p\}}\ge\epsilon$. Since users bid $(f,p)=(v_i,\epsilon)$ by the assumption that they are nonstrategic, it holds that
$\min{\{f-b_t,p\}}=\min{\{v_i-b_t,\epsilon\}}$. Hence, the inclusion requirement is satisfied if and only if $ v_i-b_t\ge\epsilon$ which implies that
$X_i=\mathbf{1}\{v_i\ge b_t+\epsilon\}$. Thus,
\[P(X_i=1\mid b_t)=P(v_i>b_t+\epsilon)=\ol{F}(b_t+\epsilon), \quad \text{for any } i=1,2,\dots,N.\]
Thus, conditional on $b_t$, the $X_i$'s are independent and identically distributed (iid) with distribution (denoted by) $X\mid b_t \sim \text{Bernoulli}(p=\ol{F}(b_t+\epsilon))$, so that $\mathbb E[X\mid b_t]=\ol{F}(b_t+\epsilon)$. Since the block capacity is upper bounded by $T$, the transactions $g_t\mid b_t$ that will get ultimately included in block $B_t$ satisfy the equality $g_t\mid b_t=\min{\Big\{T,\sum_{i=1}^{N_t} X_i\Big\}}\;\Big |\; b_t$. Putting these together, we can now upper bound $\mathbb{E}[g_t\mid b_t]$ as follows
\begin{align}\label{eq:min}
\mathbb{E}[g_t\mid b_t]&=\mathbb{E}\Big[ \min{\Big\{T,\sum_{i=1}^{N_t} X_i\Big\}}\;\Big |\; b_t\Big]\le \min{\Big\{T,\mathbb{E}\Big[\sum_{i=1}^{N_t} X_i\;\Big |\; b_t\Big]\Big\}}\nonumber\\&=
\min{\left\{T,\mathbb{E}[N_t]\mathbb{E}[X\mid b_t]\right\}}=
\min{\{T,\lambda T\ol{F}(b_t+\epsilon)\}},
\end{align}
where the inequality is due to the interchange of the minimum with the expectation and the (second to last) equality due to Wald's equation since the random variables $X_i, i=1,\dots, N_t$ are independent of $N_t$. Plugging \eqref{eq:min} in the expression for $\mathbb{E}[b_{t+1} \mid b_t]$ above, yields
\[\mathbb E[b_{t+1}\mid b_t]\le b_t+b_td\min{\{1, 2\lambda \ol{F}(b_t+\epsilon)-1\}},\]
which is the the inequality in equation \eqref{eq:expectation} as claimed.
\end{proof}

Next, we show that the dynamical system $\{b_t\}_{t\ge0}$ in \eqref{eq:fluid} has a unique fixed point which is \emph{directionally stable}. Before proceeding with the formal statement and its proof, we first define the relevant terms that we will in the subsequent theoretical analysis.

\begin{definition}[Discrete Time Dynamical System]
A one-dimensi\-onal \emph{discrete time dynamical system}, $\{b_t\}_{t\in \mathbb{N}}$, is determined by an update rule $g:\mathbb{R}\to \mathbb{R}$, so that $b_{t+1}:=g(b_t)$. We will write
\[g^t(b_0):=\underbrace{g\circ g\circ \dots \circ g}_{t-\text{times}}(x),\]
to denote the $t$-th iteration of the system, i.e., the $t$-times composition of $g$ with itself (when $t=1$, we will simply write $g$ instead of $g^1$). Accordingly, a sequence $(g^t(b_0))_{t\in\mathbb{N}}$ is a called a \emph{trajectory or orbit} of the dynamics with $b_0$ as a starting point. A point $b^*$ is called a \emph{fixed point} of the dynamics if $g(b^*)=b^*$. A common technique to show that a dynamical system
converges to a fixed point is to construct a function $\Phi: \mathbb{R}\to\mathbb{R}$ such that $\Phi(g(b) < \Phi(b)$ for any $b\in \mathbb R$ unless $b$ is a fixed point of $g$. We call $\Phi$ a \emph{Lyapunov or potential} function for $g$.
\end{definition}

\begin{definition}[Directionally Stable Fixed Point]\label{def:directional}
Let $\{b_t\}_{t\ge0}$ be a one-dimensional dynamical system determined by a function $g:\mathbb R\to\mathbb R$ and let $b^*$ be fixed point of $g$, i.e., $g(b^*)=b^*$. Then, $b^*$ is called \emph{directionally stable} for $\{b_t\}_{t\ge0}$ if for every $t\ge0$ such that $b_t\neq b^*$ it holds that $(g(b_t)-b_t)/(b_t-b^*)<0$ where $g(b_t)=b_{t+1}$ for every $t\ge0$.
\end{definition}

In other words, if a fixed point is directionally stable for a dynamical system, then the dynamical system moves to the direction of that point at every iteration. To proceed, with the formal statement that $b^*$ is directionally stable for the (non-atomic) base fee dynamics $\{b_t\}_{t\ge0}$, let $F^{-1}(p):=\inf\{x\in \mathbb {R} :F(x)\geq p\}$ denote the \emph{inverse distribution function} of $F$. Since $F$ is continuous and strictly increasing by assumption, for any $p\in [0,1]$ there exists a unique $x\in \mathbb R$ such that $F^{-1}(p)=x$. Moreover, under these conditions, $F^{-1}$ is also strictly increasing. Using this notation, we can prove \Cref{lem:fixed}.

\begin{lemma}\label{lem:fixed}
Consider the deterministic dynamical system $\{b_t\}_{t\ge0}$ with
\[b_{t+1}=b_t+b_td\min{\{1,2\lambda\ol{F}(b_t+\epsilon)-1\}}.\]
Then, $b_t$ has a unique stationary point given by
\begin{equation}\label{eq:fixed}
b^*=F^{-1}(1-1/2\lambda)-\epsilon.
\end{equation}
Moreover, $b^*$ is directionally stable for any initial condition $b_0>0$ and the dynamics $\{b_t\}_{t\ge 0}$ converge to a globally attracting $db^*$-neighborhood of $b^*$, i.e., there exists a $\bar{t}\in \mathbb N$, so that $b_t\in [(1-d)b^*,(1+d)b^*]$ for any $t>\bar{t}$.
\end{lemma}

\begin{proof}[Proof of \Cref{lem:fixed}]
Let $r_t$ denote the rate of change of $b_t$, i.e.,
\[r_t:=d\min{\{1,2\lambda\ol{F}(b_t+\epsilon)-1\}}.\]
By definition, $r_t\in[-d,d]$. The process $\{b_t\}_{t\ge0}$ becomes stationary if only if $r_t$ becomes equal to $0$. Solving the equation $r^*=0$ for $b^*$ under the assumption that $F$ is continuous and increasing (and hence invertible and with an increasing inverse, $F^{-1}$) yields the unique solution
\[b^*=F^{-1}(1-1/2\lambda)-\epsilon,\]
which is the only equilibrium candidate for the deterministic dynamical system $\{b_t\}_{t\ge0}$. Note that at $b^*$, it holds that $1/2 =\lambda\ol{F}(b^*+\epsilon)$, and hence
\begin{equation}\label{eq:ast}
\min{\{1,2\lambda\ol{F}(b^*+\epsilon)-1\}}=2\lambda\ol{F}(b^*+\epsilon)-1=0. \tag{$\ast$}
\end{equation}
To see that the the point $b^*$ is directionally stable for $\{b_t\}_{t\ge0}$, we proceed with a case discrimination on the sign of $b_t, t\ge0$. Since the dynamical system is one-dimensional, this follows from a sign analysis of $r_t$.
\begin{itemize}[leftmargin=*, itemindent=0cm]
\item[$\bullet$] $b_t<b^*$. Since $F$ is strictly increasing, it holds that $F(b_t+\epsilon)<F(b^*+\epsilon)$ for any $b_t<b^*$. Hence,
\begin{align*}
r_t&=d\min{\{1,2\lambda\ol{F}(b_t+\epsilon)-1\}}>d\min{\{1,2\lambda\ol{F}(b^*+\epsilon)-1\}}\overset{\eqref{eq:ast}}{=}0,
\end{align*}
by definition of $b^*$. Hence, $r_t>0$, whenever $b_t>b^*$.
\item[$\bullet$] $b_t>b^*$. Similarly, whenever $b_t>b^*$, it will be the case that $F(b_t+\epsilon)>F(b^*+\epsilon)$. Hence,
\begin{align*}
r_t&=d\min{\{1,2\lambda\ol{F}(b_t+\epsilon)-1\}}<d\min{\{1,2\lambda\ol{F}(b^*+\epsilon)-1\}}\overset{\eqref{eq:ast}}{=}0,
\end{align*}
where the first equality in the last line follows from the observation that $\lambda(1-F(b^*+\epsilon))<T$ by definition of $b^*$.
\end{itemize}
Thus, it remains to show that $b_t$ can only have bounded oscillations in a $db^*$ neighborhood around $b^*$, i.e., that the interval $[(1-d)b^*,(1+d)b^*]$ is globally attracting for the dynamics $\{b_t\}_{t\ge0}$. Assume that for some $t>0$, $b_t>b^*$ and $b_{t+1}<b^*$ (if $b_{t+1}>b^*$, then by the definition of directional stability, $b_{t+1}$ will be closer to $b^*$ than $b_{t}$ and the claim follows). Then, it must be the case that,
\[b_{t+1}=b_t(1+r_t)>b_t(1-d)>b^*(1-d),\]
since $r_t>-d$ by definition. Similarly, if $b_t<b^*$ and $b_{t+1}>b^*$ for some $t>0$, then it holds that
\[b_{t+1}=b_t(1+r_t)<b_t(1+d)<b^*(1+d),\]
since $r_t<d$ by definition. Thus, if $|b_{\bar{t}}-b^*|<db^*$ for some $\bar{t}>0$, it must be that $|b_t-b^*|<db^*$ for any $t>\bar{t}$. This implies that there can only be bounded oscillations around $b^*$ within the $[(1-d)b^*,(1+d)b^*]$ intervals as claimed.
\end{proof}

The next natural step is to determine conditions under which the base fee converges to this candidate equilibrium or conditions under which it does not. It is important to understand that even if the base fee remains in the bounded region specified in \Cref{lem:fixed}, it may oscillate there indefinitely (jumping from above to below the equilibrium value and vice versa) causing significant fluctuations in the block load even for stationary demand. Such an instance is given in \Cref{ex:edge}.

\begin{example}\label{ex:edge}
Let $T=1000$, and assume a fixed number of $\lambda T=3000$ arriving transactions per block with equally spaced valuations in $[200,230]$ (i.e., the valuations are not drawn from a uniform distribution but are assumed to be deterministic and linearly spaced in this case). Then, starting from $b_0=100$, the process $\{b_t\}_{t\ge0}$ has the form that is shown in \Cref{fig:edge}.
\begin{figure}[!htb]
\centering
\includegraphics[width=0.48\linewidth]{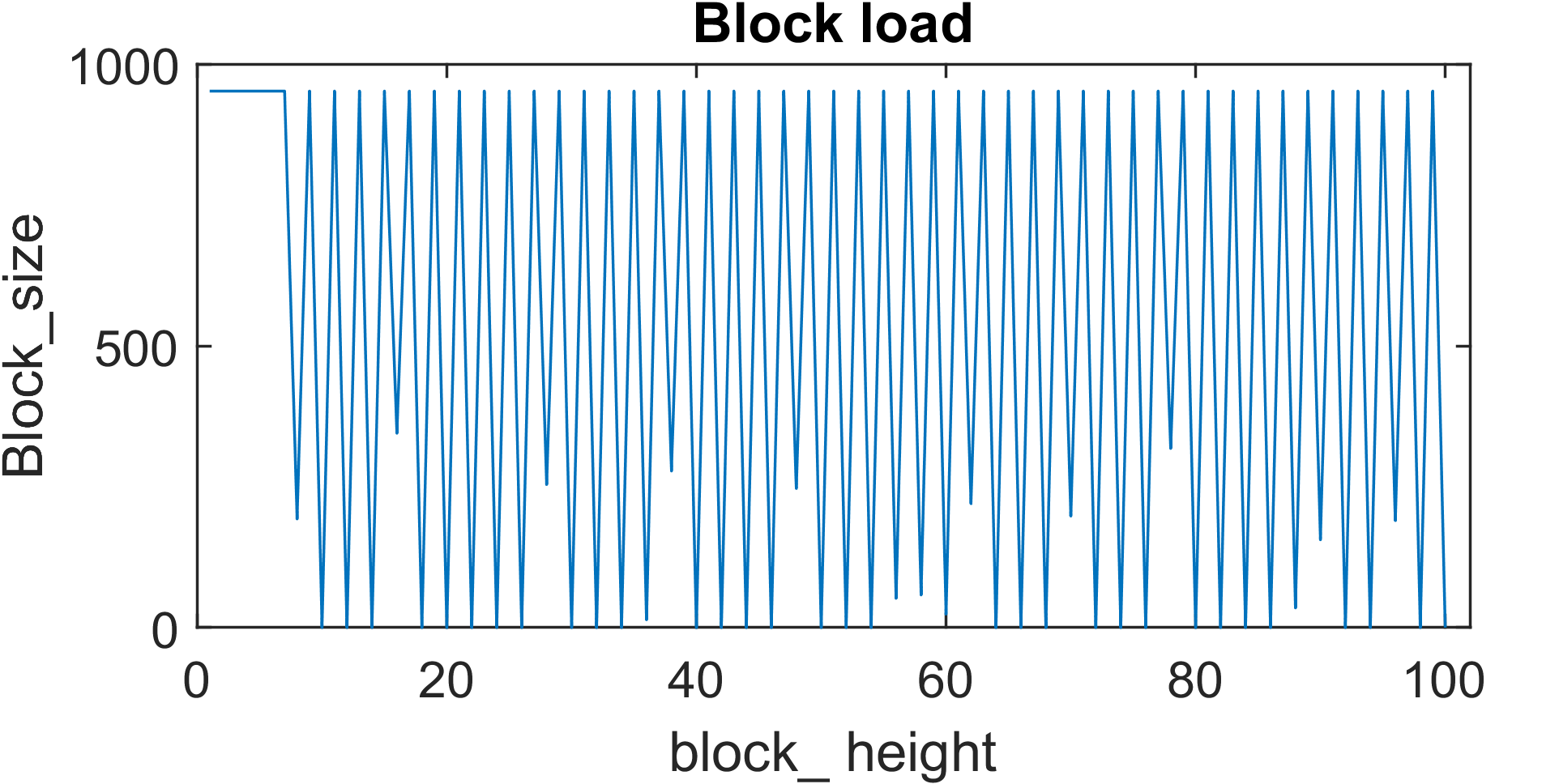}\hfill
\includegraphics[width=0.48\linewidth]{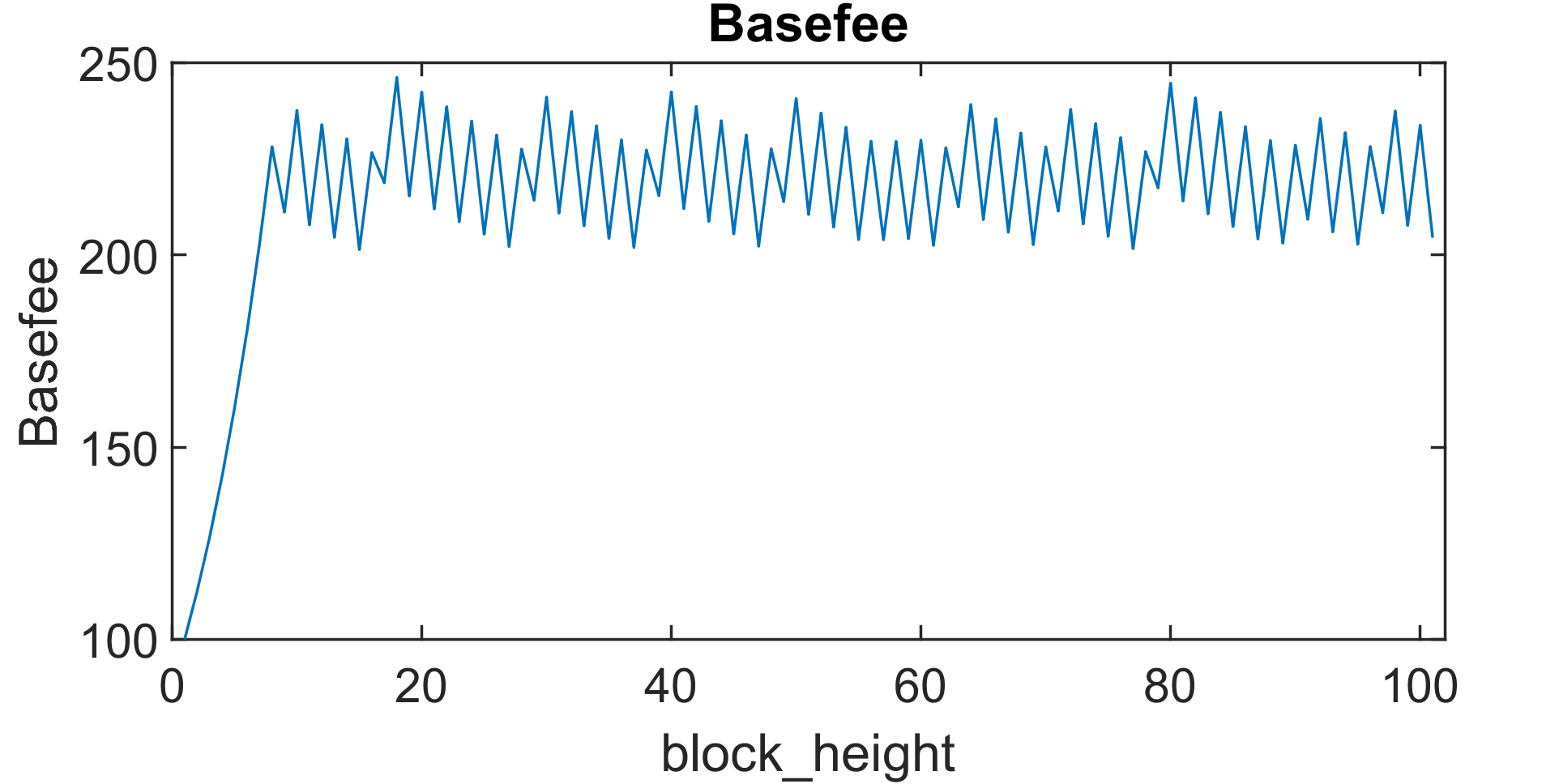}
\caption{A case with stationary demand in which the base fee, $b_t$, oscillates perpetually around the equilibrium value $b^*$ (right panel). Despite the bounded oscillations in $b_t$, the block load bounces between its extremes (full to empty and vise versa) (left panel).}
\label{fig:edge}
\end{figure}
While the base fee converges to the bounded region $[(1-d)b^*,(1+d)b^*]$ as predicted by \Cref{lem:fixed}, the block load bounces between its extremes, i.e., from full to empty (and vise versa). Intuitively, instabilities emerge as the number of arriving transactions with similar valuations increases. If valuations had significant differences, then the base fee would reach a level where only the desired $T/2$ would not be priced out. However, if valuations of users are similar, then the base fee prices out (approximately) all or (approximately) none of the users at the same time. This leads to chaotic updates of the base fee (still within the bounded region $[(1-d)b^*,(1+d)b^*]$) and as it turns out, to extreme (and undesired) oscillations in the block occupancy.
\end{example}

Our goal in the subsequent analysis is to formalize the observation in \Cref{ex:edge} and determine parameter regions for $\lambda$ and $w$ such that the base fee is provably convergent, oscillating or chaotic, leading to (approximately) stable block loads in the former case or significant fluctuations in the other cases.

\subsection{Convergence to Equilibrium}\label{sub:lower}

For lower step-sizes, we can prove convergence of the base fee dynamics to $b^*$. Here, we provide a closed form expression for the threshold under which convergence provably occurs. We remind that in the \textit{non-atomic model } the base fee $b_t$ is determined by the following dynamics
\[b_{t+1}=b_t\left [1+d\min{\{1,2\lambda\ol{F}(b_t+\epsilon)-1\}}\right].\]
Since the miners' premium, $\epsilon$ appears only in the argument of the cumulative distribution function, $F$, we will eliminate it from the following computations without loss of generality (e.g., by appropriately shifting the support of $F$). For simplicity, we assume that $\lambda=1$ so that $\min{\{1,2\lambda \ol{F}(b_t)-1\}}=2\ol{F}(b_t)-1$ for all $b_t>0$. Under these assumptions, $b^*$ simplifies to $b^*=F^{-1}(1/2)$, i.e., it is the median of the distribution $F$. Using the above, we can now formulate the following convergence threshold for the step-size (which holds for arbitrary distributions).


\begin{theorem}\label{t:threshold_convergence}
Let $b_{t+1}=b_t[1+d(2\ol{F}(b_t)-1)], t\ge0$ denote the non-atomic base fee dynamics when $\lambda=1$. Then, for any initial value $b_0>0$, and any continuous and strictly increasing distribution function, $F$, with support on $[L,U]$ with $0<L<U$, ${b_t}_{t\ge0}$ converges to $b^*=F^{-1}(1/2)$, for any step-size $d\in(0,d_F]$, where
\[d_F=\inf_{b\neq b^*}{\frac{(b^*/b)^2-1}{1-2F(b)}}\,.\]
\end{theorem}

\begin{proof}
We rewrite the base fee dynamics as \[b_{t+1}=b_t\left[1+d(1-2F(b_t)\right]\] and define the function $g:\mathbb R_+\to\mathbb{R_+}$ by $g(b):=b(1+d-2dF(b))$, for any $b>0$. We will prove that
\[(\ln g(b) - \ln b^\ast)^2 - (\ln b - \ln b^\ast)^2 < 0,\quad \text{for any } b\neq b^*.\]
Once this is established, the convergence result easily follows since $(\ln b - \ln b^\ast)^2$ acts as a potential function for the dynamics. To proceed, we rewrite the left hand side of the above inequality as
\begin{align*}
(\ln g(b) &- \ln b^\ast)^2 - (\ln b - \ln b^\ast)^2 \\&= (\ln g(b) - \ln b) \cdot (\ln g(b) + \ln b  - 2\ln b^\ast)\\
&= \ln\left(\frac{g(b)}{b} \right) \ln \left( \frac{bg(b)}{(b^\ast)^2} \right)\\
&=\ln\left[ 1 + d - 2dF(b) \right]
\cdot \ln \left[ (b/b^\ast)^2 \cdot(1 + d - 2dF(b)) \right].
\end{align*}
Since $F(b)$ is a continuous and increasing function by assumption, there are two cases:
\begin{itemize}[leftmargin=*]
\item $b<b^*$: in this case, it holds that $F(b)<F(b^*)=1/2$ which implies that $1+d-2dF(b)>1+d-2d/2=1$. Hence, $\ln[1+d-2dF(b)]>0$. Thus, to obtain the desired inequality, we need to select $d>0$ so that the term in the argument of the second $\ln$ on the right hand side of the above equation is less than $1$, i.e., $(b/b^*)^2(1+d-2dF(b))<1$. Solving for $d$ yields the inequality $d\le \frac{(b^*/b)^2-1}{1-2F(b)}$. Since this inequality must hold for any $b<b^*$, we obtain the threshold
\[d\le \inf_{b< b^*}\frac{(b^*/b)^2-1}{1-2F(b)}.\]
\item $b>b^*$: in this case, it holds that $F(b)>F(b^*)=1/2$, which implies that $\ln[1+d-2dF(b)]<0$. Thus, by a similar argument as above, we need to select $d>0$ so that the term in the argument of the second $\ln$ on the right hand side of the above equation is larger than $1$, i.e.,
\[(b/b^*)^2(1+d-2dF(b))>1.\] Solving for $d$ yields the same inequality as above (note that now $1-2F(b)<0$).\end{itemize}
Putting the two cases together, we have that the base fee dynamics converge to $b^*$ whenever $0<d\le d_f$, with $d_F=\inf_{b< b^*}\frac{(b^*/b)^2-1}{1-2F(b)}$ as claimed.
\end{proof}

We illustrate the result of \Cref{t:threshold_convergence} with an example.
\begin{example}
Consider the uniform distribution in $[0,1]$ with $F(b)=b$ for $b\in[0,1]$, $F(b)=0$ for $b<0$ and $F(b)=1$ for $b>1$. Then, $b^*=1/2$ and $d_F$ is given by $d_F=\inf_{b\neq1/2\in[0,1]}\frac{1/4b^2-1}{1-2b}$. The minimum is obtained for $b=1$ which yields the value $d_F=3/4$. This means that in this case, the dynamics converge for any $d<3/4$.\par
To see the effect of the concentration of valuations in $d_F$, consider the parametric case with $F\sim \text{Uniform}[L,U]$ with $[L,U]=[1-w/2,1+w/2]$ for some $w>0$ so that $1-w/2>0$. Then, $F(b)=(b-(1-w/2))/w$ for $b\in[1-w/2,1+w/2]$, $F(b)=0$ for $b<1-w/2$ and $F(b)=1$ for $b>1+w/2$. In this case, $b^*=1$ and $d_F$ is the solution of the optimization problem
\[d_F=\inf_{d\neq b^*\in[L,U]}\frac{(1+b)w}{2b^2},\]
which is obtained from \Cref{t:threshold_convergence} after some trivial algebra. This is decreasing in $b$ which implies that the minimum is always attained at the upper bound of the support, $b=1+w/2$, yielding the solution $d_F=\frac{w(4+w)}{(2+w)^2}$. Thus, $d_F$ is increasing in $w$ which implies that convergence is easier (harder) as valuations become less (more) concentrated in a specific regime.
\end{example}
The last example suggests that for any $d>0$, there exists a $w>0$ (small enough) such that the base fee dynamics will \emph{not} converge to $b^*$ if the valuations are uniformly distributed on an interval with range $w$. This raises the question of what happens in the base fee dynamics in such cases. As we show next, for certain values of $w$, the dynamics not only fail to converge, but they become provably chaotic.

\subsection{Instability and Chaos}\label{sub:chaos}

The previous convergence results critically depend on the provided thresholds. If the step-size exceeds these bounds, then the base fee adjustment rule may lead to chaotic updates. As mentioned above, these bounds depend on the number of arrivals, $\lambda$, and in the range of valuations, $w$. If $\lambda$ increases or $w$ decreases, i.e., if the system becomes more congested or if the valuations become more concentrated around a specific value, then the thresholds go down and a given step-size may not be enough to guarantee convergence. In fact, as we will show, for any step-size, there exists a (reasonably large) $\lambda$ and a (reasonably small) $w$ so that the dynamics become chaotic.


\subsubsection{Dynamical Systems and Li-Yorke Chaos}

Formally, we will show that the base fee updates become chaotic in the sense of Li-Yorke \cite{Li75}. If a system is Li-Yorke chaotic, then its trajectories exhibit complex behavior: uncountably many pairs of trajectories get arbitrary close and move apart infinitely many times as the system evolves. Furthermore, the system has periodic orbits of all possible periods.
This means that different trajectories become indistinguishable and hence, the system cannot be efficiently simulated or cannot be predicted in practice. The notion of Li-Yorke chaos is a fundamental notion of chaos in dynamical systems that is  connected to many other definitions of chaos (e.g., positive topological entropy). For more discussion on these connections, particularly in the case of game dynamics see~\cite{Cho21}. Such chaotic behavior has recently been observed in game theoretic settings under adaptive agents using different online learning dynamics~\cite{Pal17,Cho21,chotibut20family,bielawski2021followtheregularizedleader}. To give the formal definition of Li-Yorke chaos (cf. \Cref{def:liyorke}), we will first introduce some additional notation.

\begin{definition}[Periodic Orbits and Points]
A sequence $b_1,b_2,$ $\dots, b_k$ is called a \emph{periodic orbit} of length $k$ if $b_{t+1}=g(b_t)$ for $1\le i\le k-1$ and $g(b_k)=b_1$. Each point $b_1,b_2,\dots,b_k$ is called \emph{periodic point of period $k$}.
\end{definition}

\begin{definition}[Li-Yorke pair \cite{Li75}]
Let $X=[L,U]$ be a compact interval in $\mathbb R$ and let $g:X\to X$ define a discrete time dynamical system $(x_t)_{t\in \mathbb N}$ on $X$, so that $x_t:=g^{t}(x_0)$ for any $x_0\in X$. A pair $(x,y)\in X$ with $x\neq y$ is called a \emph{Li-Yorke pair} if
\[\liminf_{t\to\infty}|g^{t}(x)-g^t(y)|=0<\limsup_{t\to\infty}|g^t(x)-g^t(y)|.\]
If for any $x,y \in S$ with $x\neq$, the pair of $x,y$ is a Li-Yorke pair, then $S$ is called a scrambled set.
\end{definition}

The most classic definition of chaos in the mathematics literature defines chaos as the existence of periodic orbits of all possible periods along with an uncountably large scrambled set.

\begin{definition}[Li-Yorke chaos \cite{Li75}]\label{def:liyorke}
Let $X=[L,U]$ be a compact interval in $\mathbb R$ and let $g:X\to X$ define a discrete time dynamical system $(x_t)_{t\in \mathbb N}$ on $X$, so that $x_t:=g^{t}(x_0)$ for any $x_0\in X$.

The dynamical system $(x_t)_{t\in \mathbb N}$ is called \emph{Li-Yorke chaotic} if
it holds that:
\begin{enumerate}[leftmargin=*]
    \item For every $k=1,2,\dots$ there is a periodic point in $X$ with period $k$.
    \item There is an uncountable set $S\subseteq X$ (containing no periodic points), which satisfies the following conditions:
    \begin{itemize}
        \item For every $x\neq y \in S$,
        \[\limsup_{t\to\infty}|g^n(x)-g^n(y)|>0 \,\,\text{and}\,\,\liminf_{t\to\infty}|g^n(x)-g^n(y)|=0.\]
        \item For every point $x\in S$ and point $y\in X$,
        \[\limsup_{t\to\infty}|g^n(x)-g^n(y)|>0.\]
    \end{itemize}
\end{enumerate}
In particular $S$ is a scrambled set.

\end{definition}

According to \cite{Li75}, a sufficient condition for a system to be Li-Yorke chaotic is that it has a periodic orbits of period $3$.
 This will be our main tool to show that the base fee dynamics are Li-Yorke chaotic and is stated next.\par

\begin{theorem}[Period three implies chaos \cite{Li75}]\label{thm:liyorke}
Let $X\subset \mathbb R$ be a compact interval and let $g:X\to X$ be a continuous function. Further assume that there exists a point $x_0\in X$ for which the points $x_1:=g(x_0), x_2:=g(x_1)=g^{2}(x_0)$ and $x_3:=g(x_2)=g^{3}(x_0)$ satisfy
\[x_3\le x_0<x_1<x_2 \quad (\text{or }\, x_3\ge x_0 > x_1 >x_2).\]
Then, the system is Li-Yorke chaotic.

\end{theorem}

Notice that if there is a periodic point with period $3$, then the hypothesis is satisfied.

\subsubsection{Li-Yorke Chaos in the Base Fee Updates}
With the above terminology and notation at hand, we now return to the base fee dynamics. In the case of the non-atomic approximation (cf. equation \eqref{eq:fluid}), it holds that $b_{t+1}=g(b_t)$ with the continuous map $g$ defined by
\begin{equation}\label{eq:map}
    g(b)=b+bd\min{\{1,2\lambda\ol{F}(b+\epsilon)-1\}}.
\end{equation}
As we showed in \Cref{lem:fixed}, the dynamics will ultimately enter the bounded interval $X:=[(1-d)b^*,(1-d)b^*]$. Moreover, it holds that $g(b)=b$ for $b\in X$ if and only if $b=b^*$, i.e., $b^*$ is the unique fixed point of function $g$ in $X$. Thus, according to \Cref{def:liyorke} and \Cref{thm:liyorke}, it suffices to show that the continuous map $g:X\to X$ has a periodic point of period $3$, i.e., that there exists a point $b'\in X$ with $b'\neq b^*$, which is a fixed point of $g^3(b)$, i.e., $g^3(b')=b', b'\neq b^*\in X$. The two panels in \Cref{fig:liyorke} illustrate the two possible cases.

\begin{figure}[!htb]
    \centering
    \includegraphics[width=0.4\linewidth]{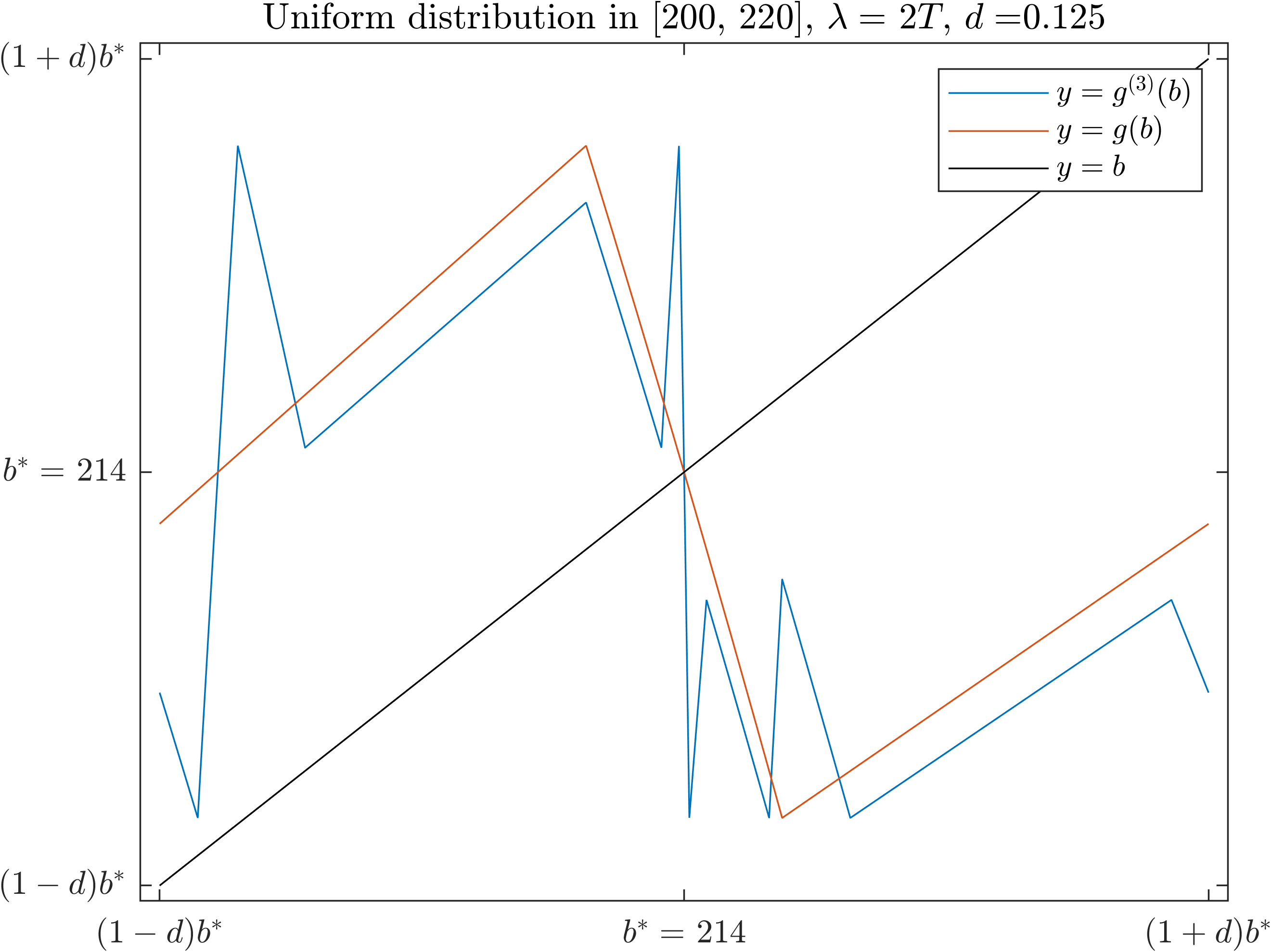}\hspace{15pt}
    \includegraphics[width=0.4\linewidth]{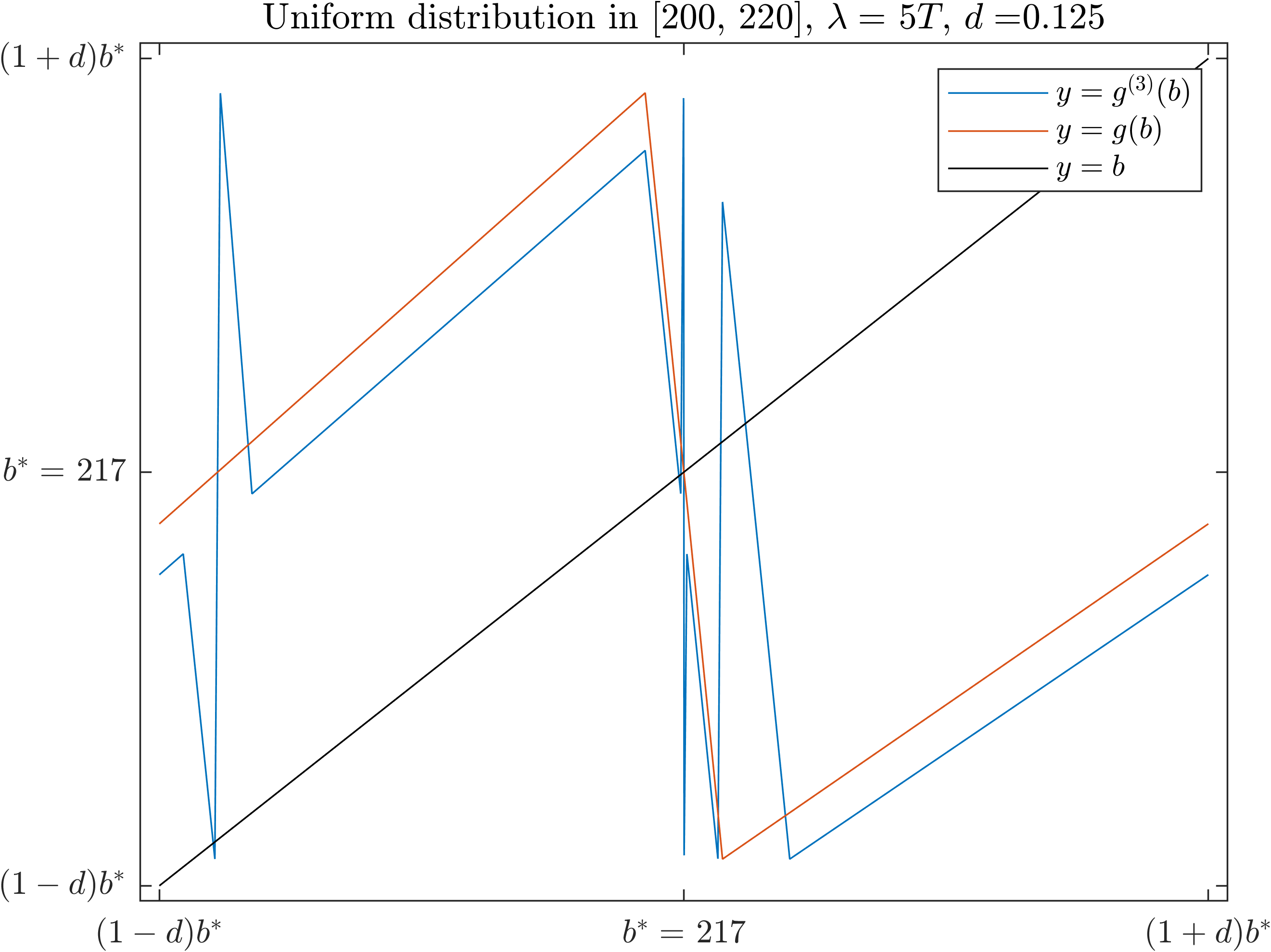}
    \caption{Orbits in the base fee dynamics $b_{t+1}=f(b_t)$. The left panel shows an instance in which the base fee dynamics do not have orbits of period 3 (the graph of $g^{(3)}(b)$ does not intersect the diagonal $y=b$, i.e., $g^{(3)}(b)$ does not have fixed points other than the unique fixed point of $g(b)$). By contrast, the right panel shows an instance with points of period 3 (multiple intersections of $g^{(3)}(b)$ and $y=b$). In this case, the dynamics are Li-Yorke chaotic. In both cases, the step size is equal to $d=0.125$ and the valuations are uniformly distributed in $[200,220]$. The difference is in the demand level which is higher in the chaotic case ($2T$ in the left panel versus $5T$ in the right panel).}
    \label{fig:liyorke}
\end{figure}

In \Cref{thm:chaos}, we invoke \Cref{thm:liyorke} and show the more general case, that for any $d$, there exists a distribution of valuations so that the system becomes chaotic.

\begin{theorem}\label{thm:chaos}
Let $g(b)=b+bd\min{\{1,2\lambda\ol{F}(b+\epsilon)-1\}}$ denote the non-atomic approximation of the update rule for the base fee dynamics $(b_{t})_{t\in \mathbb{N}}$. Then, for any fixed step size $d>0$, there exists a continuous distribution $F$ of valuations, and a point $b_0\in \mathbb R$, so that
\begin{equation}\label{eq:condition}
g^3(b_0)\le b_0<g(b_0)<g^2(b_0), \tag{PO}
\end{equation}
In particular, for any step-size $d$, there exists a distribution of valuations $F$, for which the base fee dynamics become Li-Yorke chaotic.
\end{theorem}

\begin{proof}
The proof is constructive and proceeds by creating a specific instance of the uniform distribution for which condition \eqref{eq:condition} is satisfied. Then, the claim that the dynamics are Li-Yorke chaotic follows from \Cref{thm:liyorke}. To create such an instance, let $F\sim\text{Uniform}[\mu-1/2,\mu+1/2]$ for some $\mu>0$. Also, let $\lambda=1$ and as above, assume without loss of generality that $\epsilon=0$ (e.g., by properly rescaling the distribution $F$). Based on these assumptions, it holds that $1>2\ol{F}(b)-1$ for any $b>0$ which implies that the update rule, $g$, of the non-atomic model becomes
\[g(b)=b(1+d-2dF(b)).\]
We will now show that we can construct a sequence of points $b_0, b_1=g(b_0), b_2=g^2(b_0)$ and $b_3=g^3(b_0)$ with the following properties
\begin{enumerate}[label=(\roman*)]
    \item $b_0\le\mu-1/2$,
    \item $b_1=g(b_0)=\mu-\delta$, for some $\delta>0$ sufficiently small (to be determined later),
    \item $b_2=g(b_1)\ge\mu+1/2$.
\end{enumerate}
We start by selecting an arbitrary $b_0> 0$ that satisfies property (i). Thus, it holds that $F(b_0)=0$ (since $b_0<\mu-1/2$ which is the lower bound of the support of the distribution of valuations) which implies that $b_1=g(b_0)=b_0(1+d)$. Combining this with property (ii), yields the first necessary condition, $b_0(1+d)=\mu-\delta$, or equivalently
\begin{equation}\label{eq:cond1}
     b_0=\frac{\mu-\delta}{1+d}, \quad \text{for some }\delta \in(0,\mu),\tag{$\star$}
\end{equation}
where the condition $\delta<\mu$ ensures that $b_0>0$ as assumed. Plugging this into property (i) yields the condition
\begin{equation}\label{eq:cond2}
\frac{\mu-\delta}{1+d}\le \mu-1/2 \Rightarrow 2d\mu-1-d+\delta\ge0. \tag{C1}
\end{equation}
Next, we calculate $b_2=g^2(b_0)=g(b_1)$. Since $b_1=\mu-\delta=b_0(1+d)$, we can determine $g(b_1)$ as follows
\begin{align*}
g(b_1)&=b_0(1+d)(1+d-2dF(\mu-\delta))\\
&=b_0(1+d)\left(1+d-2d\frac{\mu-\delta-(\mu-1/2)}{\mu+1/2-(\mu-1/2)}\right)\\
&=b_0(1+d)(1+d-2d(1/2-\delta))\\
&=b_0(1+d)(1+2d\delta)=(\mu-\delta)(1+2d\delta),
\end{align*}
where the last equality follows from \eqref{eq:cond1}. Thus, $b_2=g^2(b_0)=(\mu-\delta)(1+2d\delta)>b_1=g(b_0)$. Further, if property (ii) holds, i.e., if $b_2\ge\mu+1/2$, or equivalently if
\begin{equation}\label{eq:cond3}
(\mu-\delta)(1+2d\delta)\ge\mu+1/2,\tag{C2}
\end{equation}
(which gives a second necessary condition), then it holds that $F(b_2)=1$. This allows us to calculate $b_3=g^3(b_0)=g(b_2)$ as follows
\[g(b_2)=b_2(1+d-2d\cdot1)=b_2(1-d)=b_0(1+d)(1+2d\delta)(1-d).\]
Thus, $b_3=g(b_2)<b_2$ and it remains to show that $b_3\le b_0$. This yields the third necessary condition
$b_0(1+d)(1+2d\delta)(1-d)\le b_0$ or equivalently (assuming that $d<1$ as it is in practice)
\begin{equation}\label{eq:cond4}
  \delta\le \frac{d}{2(1-d^2)}\,. \tag{C3}
\end{equation}
In sum, given $d>0$, we need to select $\mu>0$ and $\delta\in(0,\mu)$ so that conditions \eqref{eq:cond2}, \eqref{eq:cond3} and  \eqref{eq:cond4} are satisfied simultaneously (note that we already used \eqref{eq:cond1} in the formulation of \eqref{eq:cond2}). This gives the system
\begin{eqnarray*}
2d\mu-1-d+\delta &\ge\,\, \, 0, & \eqref{eq:cond2}\\
d\delta\mu-d\delta^2-\delta-1/2&\ge\,\,\,0, & \eqref{eq:cond3}\\
\frac{d}{2(1-d^2)}&\ge\,\,\, \delta,  & \eqref{eq:cond4}
\end{eqnarray*}
Thus, if we select any $\delta>0$ that satisfies condition \eqref{eq:cond4}, it is immediate to see, that by selecting $\mu$ large enough, conditions \eqref{eq:cond2} and \eqref{eq:cond3} are always satisfied (since $\mu$ appears only in one term with positive scalars). Specifically, if we solve \eqref{eq:cond2} and \eqref{eq:cond3} for $\mu$, we obtain the (always feasible) condition
\begin{equation}\label{eq:cond5}
\mu\ge \max{\left\{\frac{1+d-\delta}{2d},\frac{d\delta^2+\delta+1/2}{\delta d}\right\}},  \tag{C4}
\end{equation}
which together with the selected $\delta$, yields an admissible solution of the initial system. Note that the second term inside the $\max$ is always larger than $\delta$, since $\frac{d\delta^2+\delta+1/2}{\delta d}=\delta+\frac{1}{d}+\frac{1}{2\delta d}>\delta$ which implies that any $\mu$ that satisfies condition \eqref{eq:cond5} immediately satisfies $\mu>\delta$ as required for $b_0$ to be positive (and thus for $b_0$ to be less than $b_1$). In sum, we have shown that if we select a point $b_0\le \mu-\delta$ where $\mu,\delta$ satisfy conditions \eqref{eq:cond4} and \eqref{eq:cond5}, then it holds that $g^3(b_0)\le b_0 < g(b_0)<g^2(b_0)$, which concludes the proof.
\end{proof}
Note that the construction in the proof of \Cref{thm:liyorke} was based in a \emph{favorable} scenario for stability which assumed $\lambda=1$. For higher values of $\lambda$, the construction still applies and in fact, chaos obtains for a much wider range of parameters (see \Cref{sub:bifurcation}). Moreover, the selection of the uniform distribution in the proof is not binding and the proof idea applies for arbitrary distributions. This is illustrated in the next example which concludes this section.

\begin{example}[A Specific Instance with Period 3]
Let $b_0>0$ take an arbitrary value and assume that the cumulative distribution function, $F(b)$ of the valuations is continuous and satisfies the following conditions: $F(0.64 b_0) = 6/32$, $F(0.8b_0) = 11/18$, $F(b_0) = 11/18$ and $F(1) = 1$. Assume that $\lambda=1$ and that $F$ is rescaled so that $\epsilon=0$ (as above). Then, for $d = 9/10$, we have that $b_1 = g(b_0)=b_0(1  +(9/10)(2(1-7/18)-1))= 0.8b_0$, $b_2=g^2(b_0)=0.8b_0(1+(9/10)(2(1-11/18)-1)=0.64b_0$ and $b_3=g^3(b_0)=0.64b_0(1+(9/10)(2(1-6/32)-1)=b_0$. We, thus, get an example with period $3$.
\end{example}

\subsection{Bifurcation Diagrams}\label{sub:bifurcation}

The previous paragraphs suggest that there are ranges of parameters for which the base fee dynamics converge and ranges of parameters for which they become Li-Yorke chaotic. The system is more prone to chaotic behavior as the step-size, demand (users that submit transactions) or concentration of valuations increase. In this paragraph, we visualize the long-term behavior of the base fee dynamics and the transitions through the various regimes as the critical input parameters of the system change. Again, for expositional purposes, we restrict attention to uniform distribution of valuations on the interval $[L,U]=[210-w/2,210+w/2]$.\footnote{Simulations with different distributions such as triangle distribution or normal produce qualitatively equivalent results which are not presented here.} The results are shown in the \emph{bifurcation diagrams} in \Cref{fig:bifwld}.
\begin{figure}
    \centering
    \includegraphics[width=0.48\linewidth]{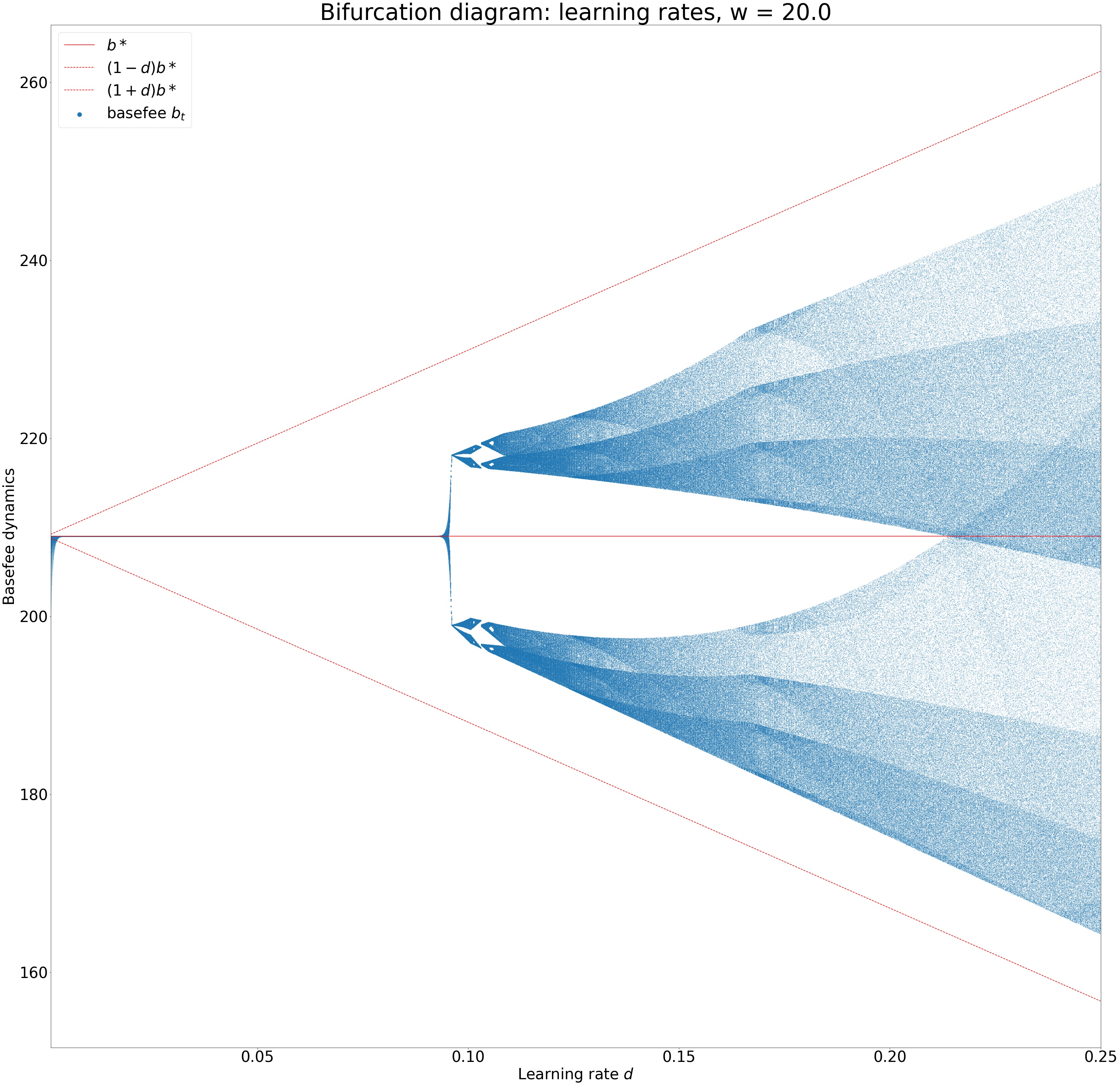}\hfill
    \includegraphics[width=0.48\linewidth]{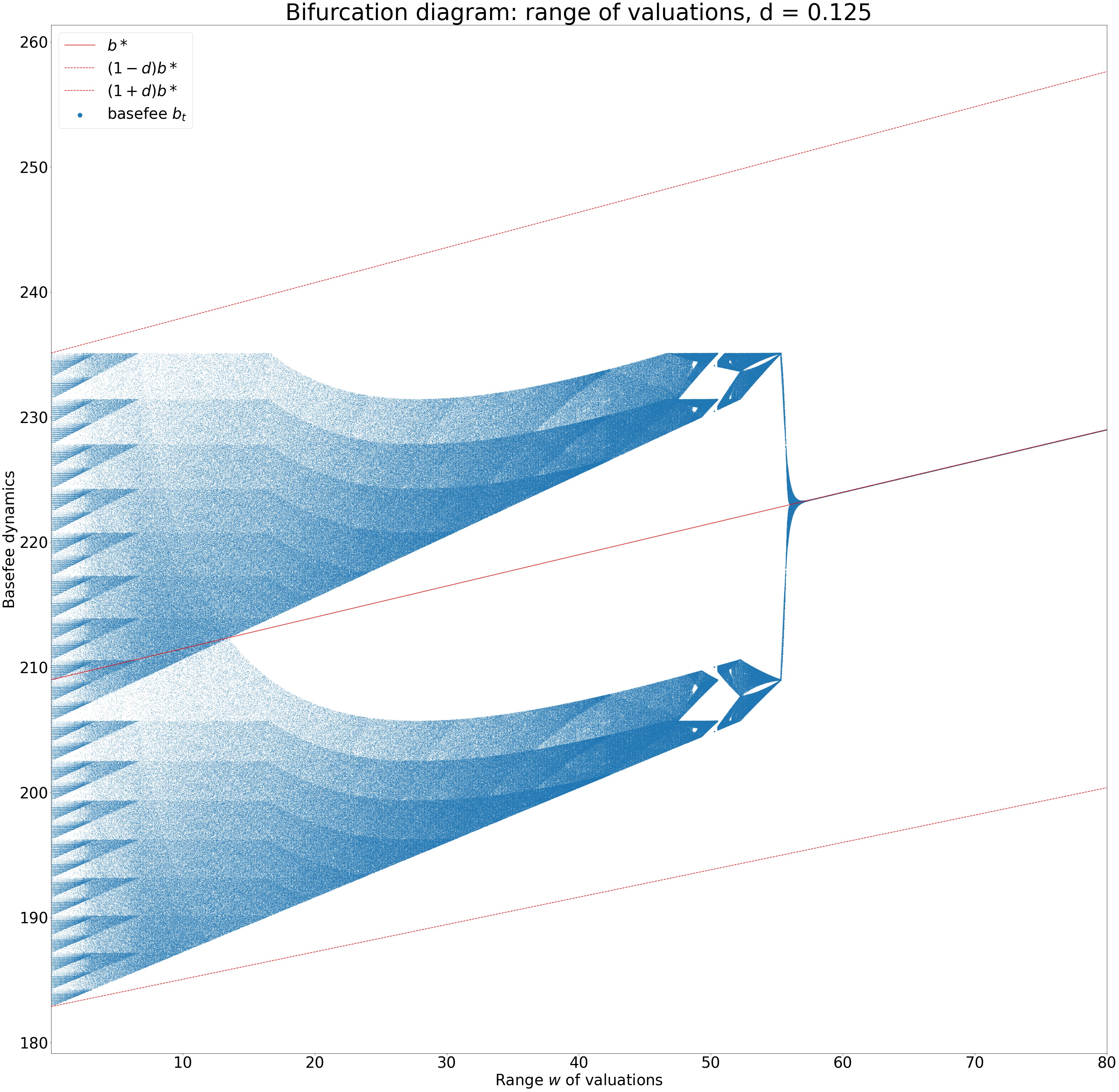}
    \caption{Bifurcation diagrams for the input parameters, $w$ (range of valuations) and $d$ (step-size) of the non-atomic base fee dynamics of equation \eqref{eq:fluid}. Left panel: route from order to chaos as the step-size increases. Right panel: route from chaos to order as the range of valuations increases.}
    \label{fig:bifwld}
\end{figure}

The horizontal axis of each diagram corresponds to the varying parameter, $w$ and $d$ respectively, with all other parameters being fixed.\footnote{The bifurcation diagram for parameter $\lambda$ is similar and is not presented here.} The vertical axis shows the attractor of the base fee dynamics (blue dots) for 400 updates (after a burn-in period of 100 updates) and the $[(1-d)b^*,b^*,(1+d)b^*]$ bounds. Interestingly, the transition from the stable (convergent) to the chaotic regime does \emph{not} occur by a period doubling (as is typical in most game-theoretic applications of chaos theory) \cite{Cho21,bielawski2021followtheregularizedleader}. For practical purposes, the important observation is that these phase transitions occur abruptly for small changes in the parameter values.




%% file: experiments.tex
\section{Experiments}
\label{sec:experiments}

We describe here the main components and results from a simulation environment created to replicate the Ethereum transaction fee market.

\subsection{Simulation environment}

\subsubsection{Chain dynamics}

Blocks in Ethereum are produced by miners in a random, iterative process. A block builds on a chain of predecessors, such that the chain length always increases in time. The consensus algorithm ensures all participants (miners and users) agree on the current head of the chain. In the following simulations, we adopt the parameter choices of EIP 1559, namely, a gas target of 12,500,000, gas limit of 25,000,000 and update rate parameter $d = 0.125$.

The simulations assume that a unique block is produced at all chain heights and that all participants receive this block with no latency. In particular, once a block is published, all users observe the updated basefee instantly. We assume randomness in the block arrival times. In Ethereum, block arrival times typically follow a Poisson process of rate $\eta = 13$ seconds under the assumption of no latency.\footnote{Average block time chart \url{https://etherscan.io/chart/blocktime}}

\subsubsection{User behavior}

User values $v$ are sampled from a fixed distribution $F$. We assume users all send transactions consuming the same amount of gas $\gamma$, without loss of generality. The values are expressed as benefit received per unit of gas, thus if the user's transaction is included, the user receives $\gamma v$ utility. The parameter value $\gamma$ is obtained by computing the average gas used in all blocks from a sample period between block 10,900,000 (timestamp Sep-20-2020 03:17:06 PM UTC) to block 10,942,000 (timestamp Sep-27-2020 02:40:14 AM UTC). We then take the median of the series of average gas used per block, rounded to the nearest multiple of 1,000, to provide a sensible $\gamma$ estimate. The procedure yields $\gamma = 76,000$, for a maximum number of transactions in the block equal to 328, given the gas limit of 25,000,000. User values are chosen uniformly at random in the interval from 10 to 100 Gwei.

Users typically transact on the Ethereum chain through \textit{wallets}, which provide fee estimation and generate transaction parameters such as the gas limit and data payload (e.g., inputs and names of function calls for smart contract interactions). In the current, first-price auction-based fee market, wallets typically provide fee estimation by computing statistics from historical transaction inclusion, e.g., the median fee paid by any transaction over the last 200 blocks. While wallet behaviors in 1559 are not currently known, we use the following design:

\begin{itemize}[leftmargin=*]
    \item When the previous block was not close to full (less than 90\% gas used), the wallet sets the fee cap parameter $f$ to a fixed value derived from the basefee (we use three times the current basefee) and the premium $p$ to the commonly agreed $\epsilon$ miner marginal cost. We guarantee that fee cap covers at least the premium by setting a lower bound.
    \item When the previous block was close to full (above 90\% gas used), the wallet adds an increment to the average tip recorded in the previous block, inducing competition between users while basefee matches the new demand.
\end{itemize}

\subsubsection{Demand process}

For convenience, we introduce two time indices: $s, t$ refer to chain heights (measured in blocks), while $k$ refers to simulation time (measured in seconds). As block inter-arrival times are random, we first generate a demand process $(D_k)_k$ returning for all time indices $k$ an integer-valued demand volume. $D_k$ is interpreted as ``users producing transactions between seconds $k-1$ and $k$''.

To generate $(D_k)_k$, we sample Brownian motion (BM) paths with initial condition $D_0$, mean 0 and variance $\sigma^2$. We obtain demand paths that feature periods of increasing and decreasing volumes due to the randomness of the BM, yet do not explicitly have a positive or negative trend. In addition, we simulate random ``jumps'' where a mass of users is generated at random intervals, decaying over the next steps at a rate $\delta$, to reproduce instances where an on-chain event brings a sudden influx of new users (e.g., token sale). Formally, our demand process satisfies at time $k$:
$$D_k = W_k + J_k; \quad J_k = (1-\delta) J_{k-1} + \sum_{j = 1}^{M_k} \zeta_j \cdot \textbf{1}_{\lfloor \kappa_j \rfloor = k}; \quad J_0 = 0$$
where $W_k$ is a discretized Brownian motion (in this case, a random walk with normal increments of mean 0 and standard deviation $\sigma$); $M_k$ is a Poisson process of rate $\lambda_j$ evaluated at time $k$; $\zeta_j$ is an exponential random variable of mean $B_0$, modelling a demand jump; and $\lfloor \kappa_j \rfloor$ is the time index where the $j$-th jump occurred. Figure~\ref{fig:bmpaths} depicts the sample paths for one value of demand variance and initial condition.

\subsubsection{Transaction pool behavior}

Miners run Ethereum nodes exchanging data (including transactions) with other nodes over a peer-to-peer network. Nodes are either run by miners, users or third parties, to relay this data in a decentralized fashion. A user either directly sends their transactions from their own node or indirectly from a third party node, who receives the transaction from the user via some communication protocol.

While only miner nodes eventually produce blocks, all nodes feature a transaction pool that holds pending transactions and continuously receives or sends items to other nodes, as requested. All nodes are free to decide in practice which transaction pool policies to apply, including the choice of the maximum number of pending transactions held in the transaction pool at any point in time. Geth, Ethereum's dominant node client as of February 2021, holds by default a maximum of 4096 transactions in the pool.\footnote{See \url{https://geth.ethereum.org/docs/interface/command-line-options} for defaults, \url{https://www.ethernodes.org/} for client statistics.}

In our simulations, we abstract this peer-to-peer network of transaction pools with differing policies into one logical transaction pool, which instantly receives all user transactions, applies the same pool policy at all time steps and is used by all miners to form their blocks. The possibility that various pool policies will affect transaction transmission is not considered.

\subsubsection{Miner behavior}

Much like the logical transaction pool described above, we also consider a single logical miner representing all miners who produce blocks. \cite{roughgarden2020transaction} provides evidence for the incentive - compatibility of miner myopic strategies, who maximize greedily the available fees at the time of their block production. Thus, we do not consider the possibility of long-range attacks in our simulations, where a cartel of miners colludes to lower the basefee to zero to enforce a monopolistic price of entry. Given the set of pending transactions in the pool, miners order transactions by received tip (in decreasing order) and include as many transactions as possible, until the block limit is reached or there are no more valid transactions to include.

\subsubsection{Simulation steps}

We provide below a description of a single simulation step, articulating how the various components are employed.

\begin{enumerate}[leftmargin=*]
    \item The previous block $B_{t-1}$ produced at chain height $t-1$ is received by all participants.
    \item An inter-arrival time $\eta_t$ is sampled from an exponential distribution of mean $\eta = 13$, such that the block at height $t$ is created at time index $\theta_t = \sum_{s = 1}^{t} \eta_s$.
    \item Given the demand process $(D_k)_k$, where $k$ is an index over seconds, we obtain $N_t = \sum_{k=\theta_{t-1}}^{\theta_t} D_k$. $N_t$ is the number of users entering the market between blocks $t-1$ and $t$, included at the earliest in block $B_t$.
    \item The $N_t$ users observe the current chain state (e.g., the basefee level) and decide whether or not to transact. User transactions are formed via \textit{wallets} which encode shared strategies.
    \item Transactions are received by miner transaction pools, which hold a set of pending transactions from previous simulation steps. All the while, transaction pools apply eviction policies in order to manage their limited resources.
    \item The miner producing $B_t$ selects transactions from the pool to maximize their fees. The set of selected transactions must be smaller than the block limit.
    \item Repeat from step 1.
\end{enumerate}

\begin{figure*}[t]
    \centering
    \includegraphics[width=\linewidth]{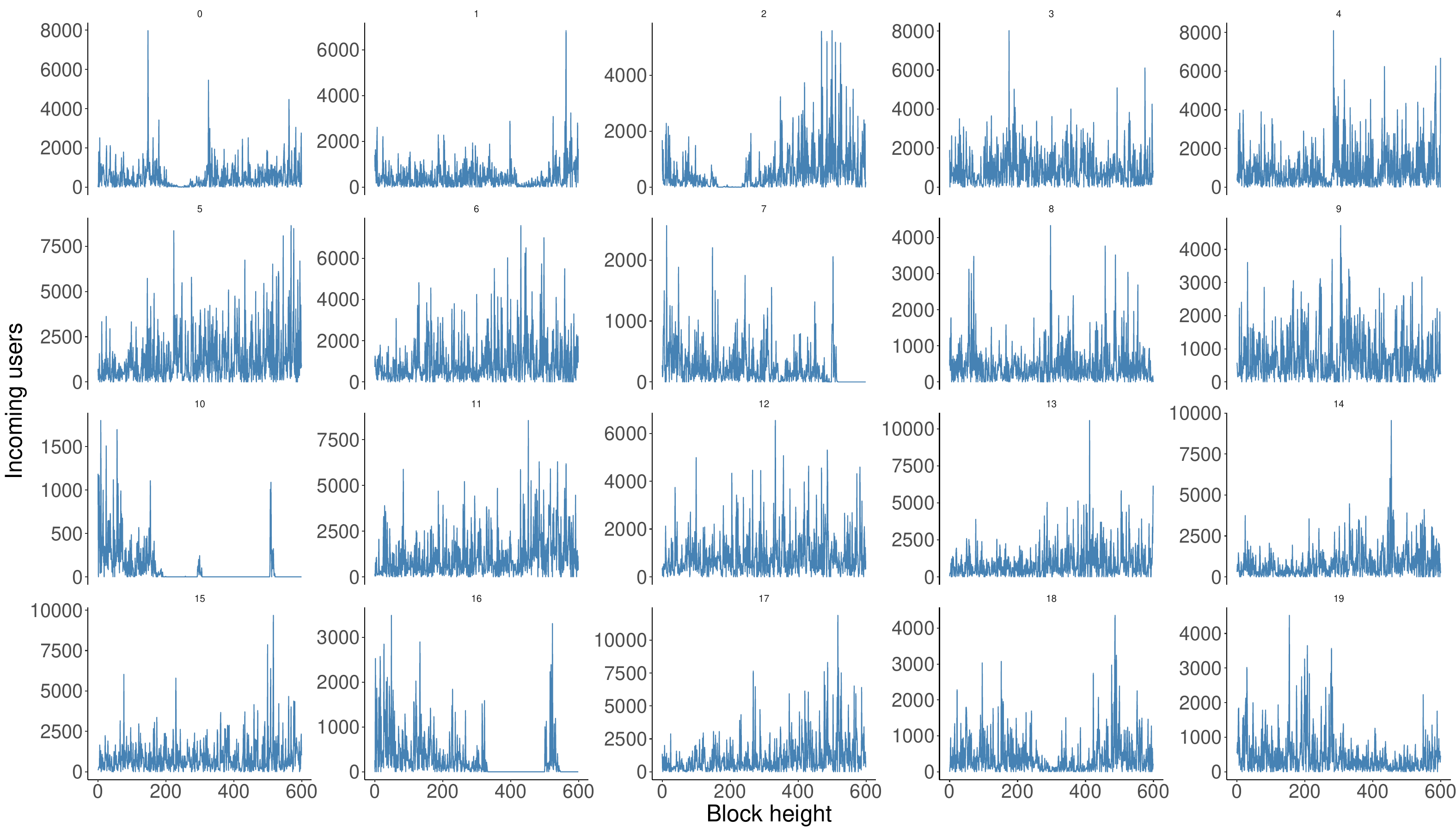}
    \caption{Demand paths at standard deviation $\sigma = 0.5$, $D_0 = 2 \frac{T}{\gamma \eta}$.}
    \label{fig:bmpaths}
\end{figure*}

\subsection{Simulation results}

\subsubsection{Independent and dependent variables}

We focus our attention on three independent variables:

\begin{itemize}[leftmargin=*]
    \item The demand path variance $\sigma^2$: We choose $\sigma \in \{ 0.1, 0.5, 1.0 \}$.
    \item Initial condition of the Brownian motion $D_0$: Given the mean block inter-arrival time $\eta$, we select $D_0$ to reproduce conditions of low, medium and high demand. With $D_0$ set to $\frac{T}{\gamma \eta}$, we target a user arrival rate that on average is exactly enough to fill the block to its limit (i.e., to twice its target). Thus we select $D_0 \in \{ \frac{T}{\gamma \eta}, 1.5 \frac{T}{\gamma \eta}, 2 \frac{T}{\gamma \eta} \}$.
    \item Transaction pool eviction tolerance $\tau$: size of the eviction band, i.e., evict all transactions with fee cap smaller than $(1-\tau)$ times basefee. $\tau = 0$ is the strictest policy (remove all transactions with fee cap smaller than current basefee), $\tau = 1$ is the most permissive (keep everything, modulo the pool limit size). We choose $\tau \in \{ 0, 1/3, 2/3, 1 \}$.
\end{itemize}

We sample 20 Brownian motion samples (see Figure~\ref{fig:bmpaths}). Each sample yields nine distinct paths, one for each value of standard deviation $\sigma$ and initial condition $D_0$, i.e., 180 paths. For each path, one simulation is run for each value of $\tau$, yielding 720 sample runs. Each run consists of 600 blocks, representing approximately half a day of activity on Ethereum. The first 100 blocks of each run are discarded from the analysis, as they represent initial conditions where basefee has not yet matched the existing demand.

Our dependent variable measures the variance of recent realisations of the percentage of gas used by the block. In Section~\ref{sec:analysis}, chaotic behavior obtained rapid variations of the gas used, from mostly empty block to mostly full blocks. Experimentally, we measure a moving standard deviation of the series of gas used, with window size 4. The maximal standard deviation $s^*$ is achieved whenever the four values alternate between 0 and 100. We call \textit{high variance} time steps where the moving standard deviation is at least 95\% of $s^*$. The percentage of high variance time steps among all simulation steps is our dependent variable.

\subsubsection{Higher variance, higher initial demand and more permissive pool policies increase high variance periods} We reproduce in Figure~\ref{fig:highvar} the percentage of high variance steps averaged over sample runs for each single value of treatment variables. We observe consistent increases in high variance steps as the demand variance increases, the initial demand level increases or the pool eviction policy is more permissive, as evidenced by the two samples presented in Figure~\ref{fig:sample}.

The pool eviction has the sharpest contrast between levels of the variable. While high variance steps almost never occur with the strictest pool policy (never keeping a transaction with fee cap inferior to the current basefee), even a mild increase of the tolerance to 1/3 induces a level of high variance steps comparable to any further increase of the tolerance.

\begin{figure}
  \centering
  \begin{subfigure}{.3\linewidth}
    \centering
    \includegraphics[width = \linewidth]{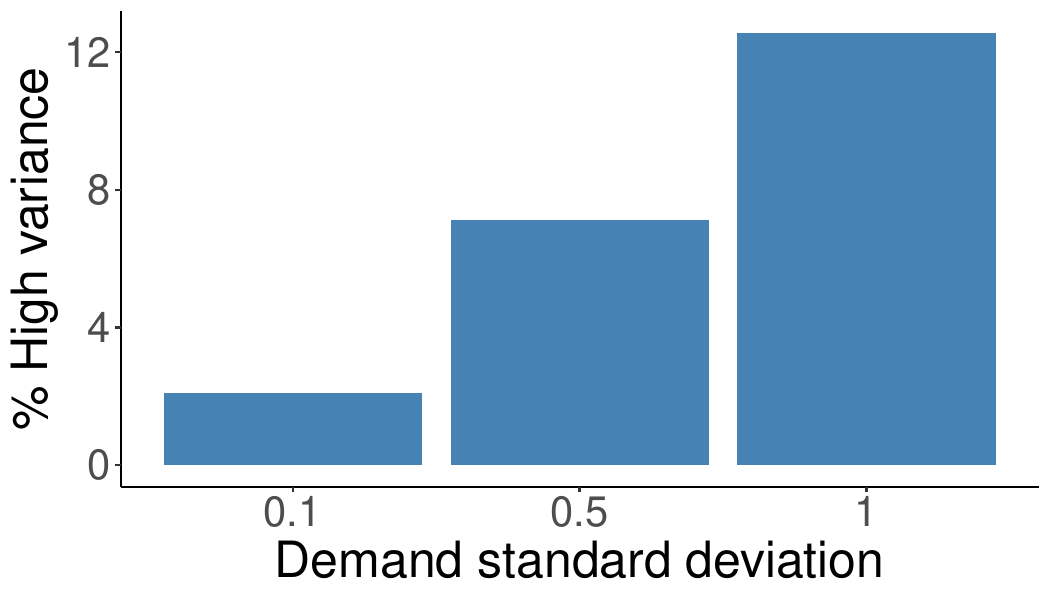}
  \end{subfigure}
  \hfill
  \begin{subfigure}{.3\linewidth}
    \centering
    \includegraphics[width = \linewidth]{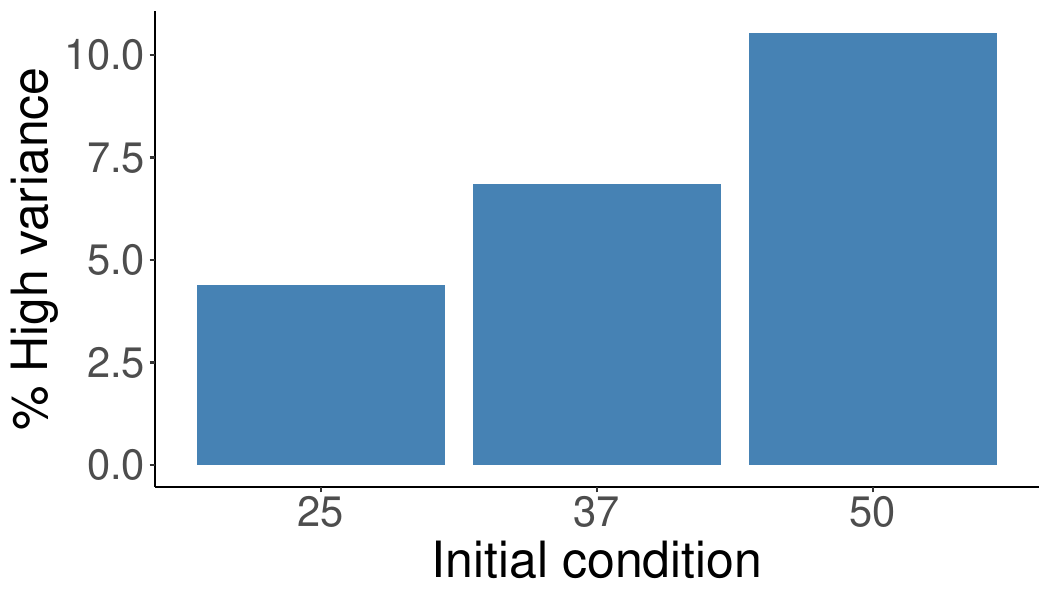}
  \end{subfigure}
  \hfill
  \begin{subfigure}{.3\linewidth}
    \centering
    \includegraphics[width = \linewidth]{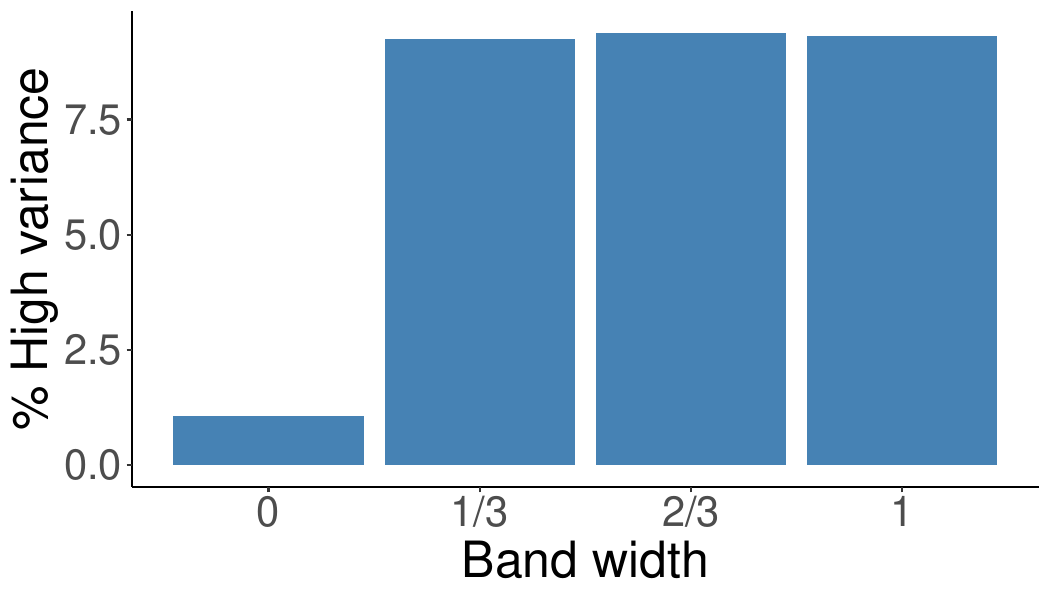}
  \end{subfigure}
  \caption{Percentage of high variance simulation steps with changing demand variance, initial condition and transaction pool eviction tolerance.}
  \label{fig:highvar}
\end{figure}

\subsubsection{Stricter pool policies trade-off user efficiency and miner revenue} In the simulations, user values are randomly sampled from the uniform distribution, with the randomness seeded by the index of the Brownian motion (BM) sample, such that runs with identical demand paths, demand variance and initial condition generate the same users. This allows for comparison of two more dependent variables, \textit{user efficiency} and \textit{miner revenue}, given the band width $\tau$ as independent variable and controlling for demand variance and initial conditions.

User efficiency measures the total benefits received by all users included in the chain. EIP 1559 is efficient whenever users with the highest value are included on-chain.\footnote{In this work, we do not consider users with time preferences.} Miner revenue consists of the received tips, either at the marginal cost level 1 Gwei per gas unit or higher whenever users are competitively bidding.

An experiment is represented by a choice of triple (BM index, demand variance, initial condition), with the band width $\tau$ taken as the independent variable. In all experiments, increasing the band width decreases both the user efficiency as well as miner revenue. This result is explained by the dynamics of the pool itself. By keeping transactions that are not currently valid for inclusion, miners have ``inventory'' to spend whenever demand is low and the basefee has decreased enough to make the transactions valid. This inventory however represents a danger to the stability of the basefee, as a bottleneck of transactions may accumulate in the pool, all becoming valid at the same instant and provoking basefee spikes and instability.

\begin{figure}
  \centering
  \begin{subfigure}{.3\linewidth}
    \centering
    \includegraphics[width = \linewidth]{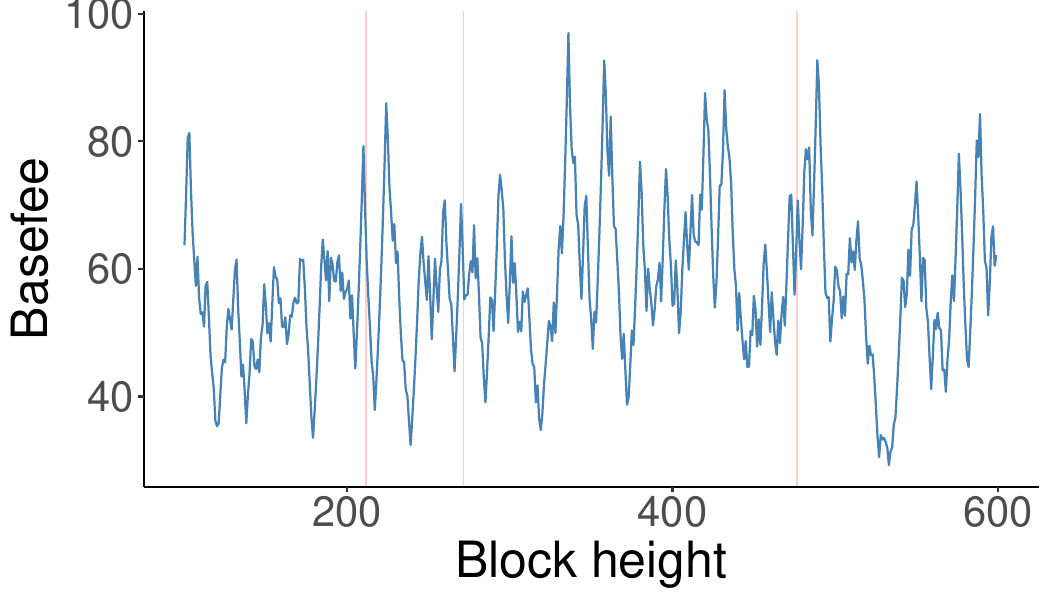}
  \end{subfigure}
  \hfill
  \begin{subfigure}{.3\linewidth}
    \centering
    \includegraphics[width = \linewidth]{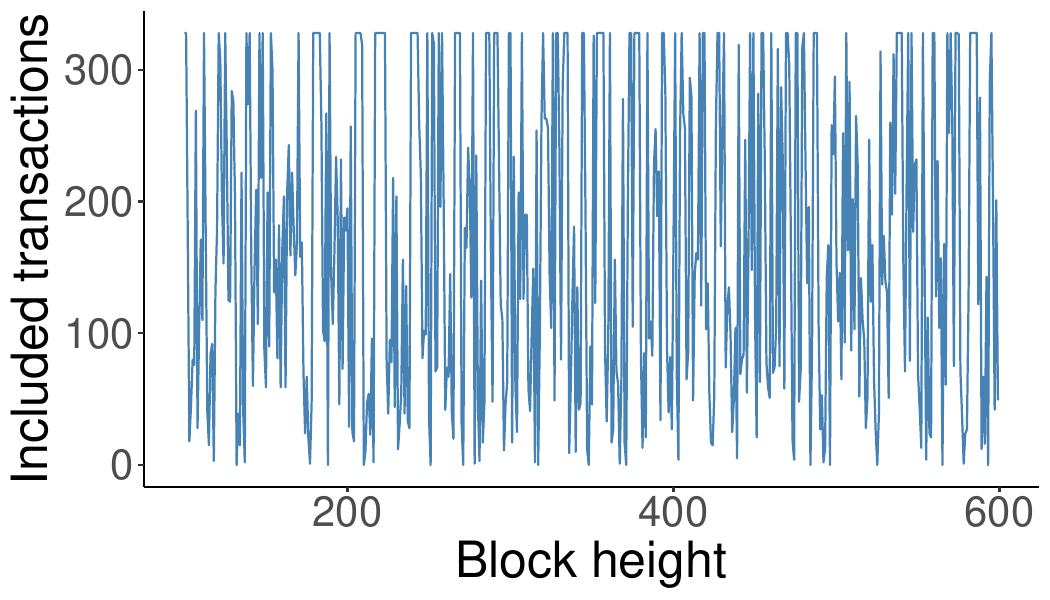}
  \end{subfigure}
  \hfill
  \begin{subfigure}{.3\linewidth}
    \centering
    \includegraphics[width = \linewidth]{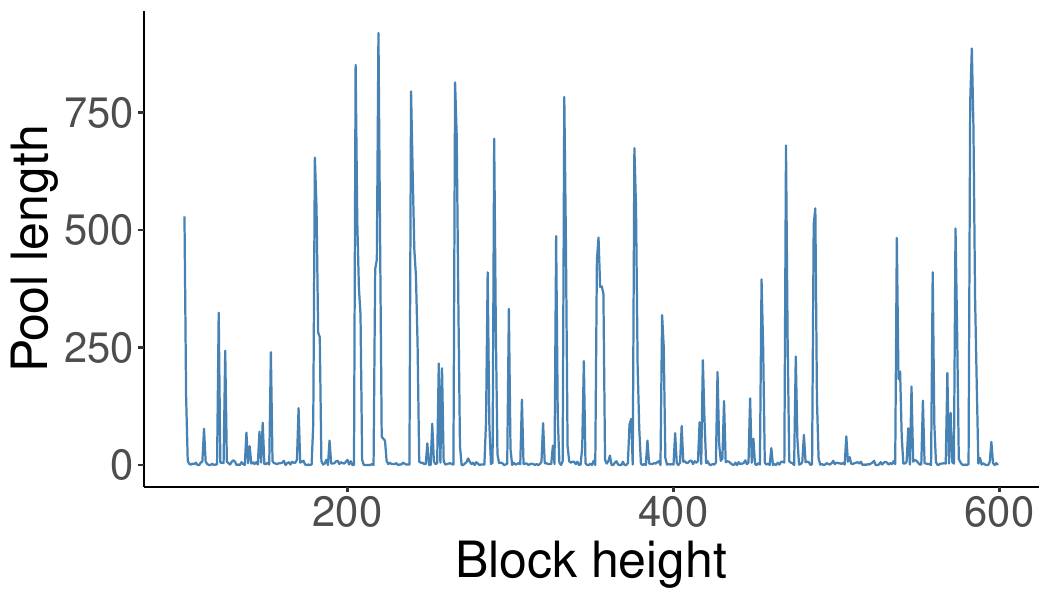}
  \end{subfigure}\\
  \begin{subfigure}{.3\linewidth}
    \centering
    \includegraphics[width = \linewidth]{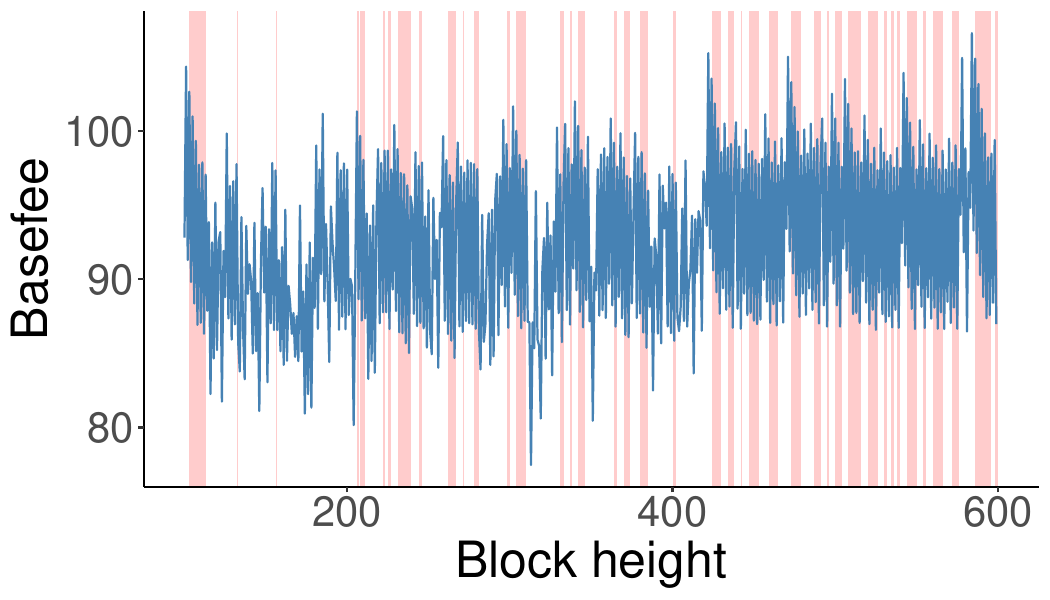}
  \end{subfigure}
  \hfill
  \begin{subfigure}{.3\linewidth}
    \centering
    \includegraphics[width = \linewidth]{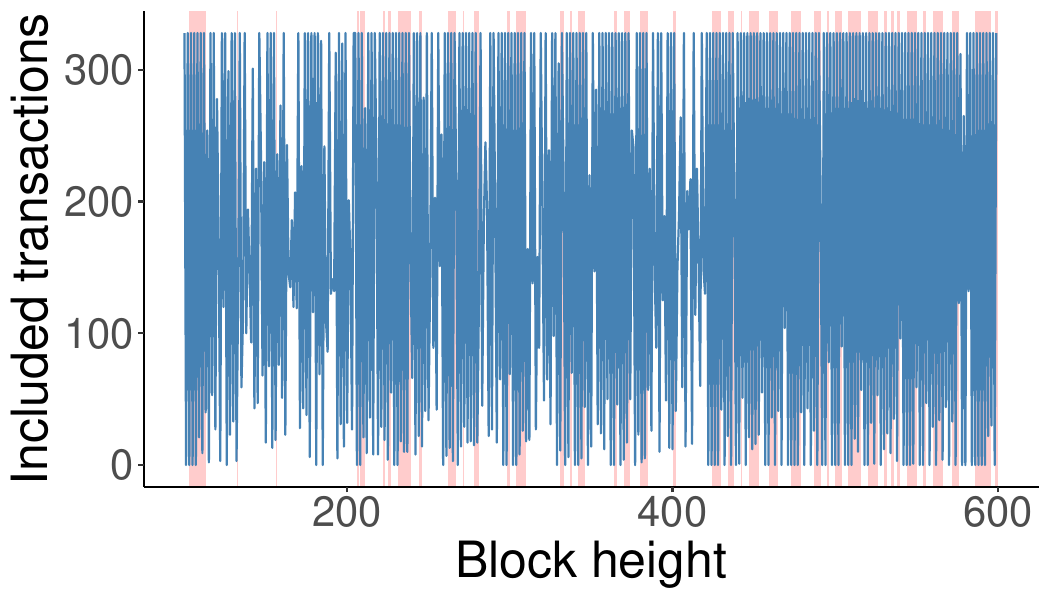}
  \end{subfigure}
  \hfill
  \begin{subfigure}{.3\linewidth}
    \centering
    \includegraphics[width = \linewidth]{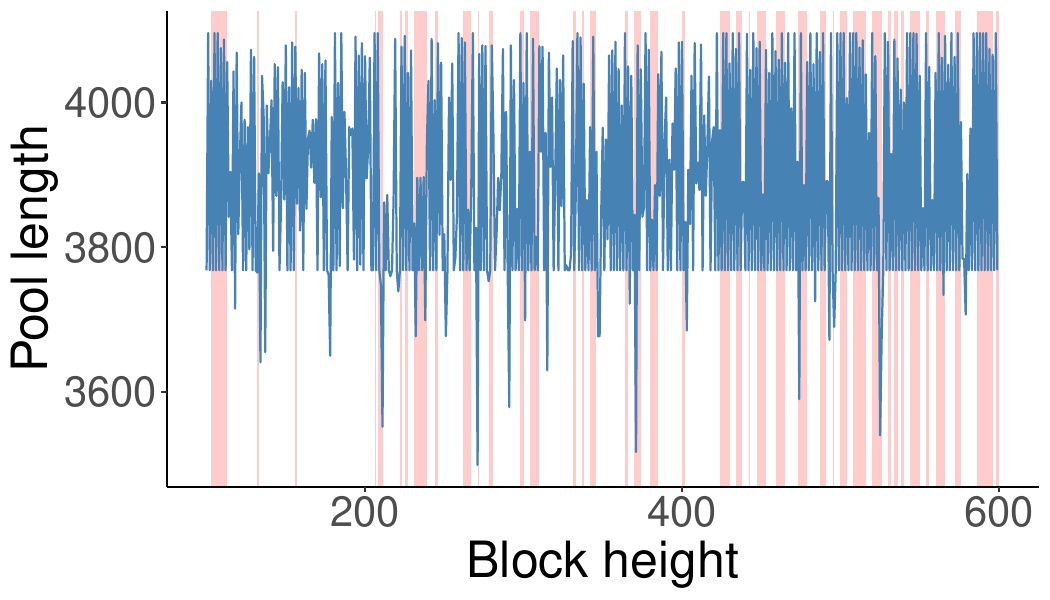}
  \end{subfigure}
  \caption{Two sample runs, one per row, with line plots for basefee, number of included transactions and transaction pool length. High variance periods are represented by red bars in the background. The first sample has $D_0 = \frac{T}{\gamma \eta}$, $\sigma = 0.1$, $\tau = 0$ and features few high variance periods. The second sample has $D_0 = 2\frac{T}{\gamma \eta}$, $\sigma = 1$, $\tau = 1$ and features more high variance periods. Additionally, the transaction pool is continuously full, as the pool eviction policy is most permissive.}
  \label{fig:sample}
\end{figure}

%% file: conclusions.tex
\section{Conclusions}\label{sec:conclusions}

Ethereum's improvement proposal (EIP) 1559 is aiming to transform the transaction fee market of the Ethereum blockchain via a dynamic pricing mechanism. The core element of the mechanism is a fixed-per-block network fee (termed basefee) that is burned and dynamically expands/contracts block sizes to deal with transient congestion \cite{Con19}. Our goal in this paper was to stress-test the basefee both theoretically and experimentally and understand its effects on regulating the transaction fees and block occupancies.\par
A concrete outcome of both our theoretical and experimental analysis is the importance of the basefee adjustment parameter in the performance of the mechanism. Our findings provide insights about the conditions under which the basefee self-stabilizes but also characterize extreme operational scenarios under which its dynamics become chaotic. In particular, we showed that EIP1559 has promising properties (convergence guarantees under various conditions) to convey stability to the fee market and identified sources of concern that may destabilize the system into regimes of chaotic behavior. Our work develops a systematic framework that combines elements from mechanism design, dynamical systems and chaos theory and which aims to aid the ongoing study of transaction fee markets in blockchain-based economies. 

%% file: arxiv.bbl
\begin{thebibliography}{10}

\bibitem{bachrach2010honor}
Y.~Bachrach.
\newblock Honor among thieves: Collusion in multi-unit auctions.
\newblock In {\em Proceedings of the 9th International Conference on Autonomous
  Agents and Multiagent Systems: Volume 1 - Volume 1}, AAMAS '10, page
  617–624, Richland, SC, 2010. International Foundation for Autonomous Agents
  and Multiagent Systems.

\bibitem{bielawski2021followtheregularizedleader}
J.~Bielawski, T.~Chotibut, F.~Falniowski, G.~Kosiorowski, M.~Misiurewicz, and
  G.~Piliouras.
\newblock {Follow the Regularized Leader Routes to Chaos in Routing Games},
  2021.

\bibitem{But19}
V.~{Buterin}, D.~{Reijsbergen}, S.~{Leonardos}, and G.~{Piliouras}.
\newblock Incentives in ethereum’s hybrid casper protocol.
\newblock In {\em 2019 IEEE International Conference on Blockchain and
  Cryptocurrency (ICBC)}, pages 236--244. IEEE, USA, May 2019.

\bibitem{Con19}
Vitalik Buterin, Eric Conner, Rick Dudley, Matthew Slipper, Ian Norden, and
  Abdelhamid Bakhta.
\newblock {EIP1559: Fee market change for ETH 1.0 chain}.
\newblock \href{https://eips.ethereum.org/EIPS/eip-1559}{Published online},
  2019.

\bibitem{Car15}
I.~Caragiannis, C.~Kaklamanis, P.~Kanellopoulos, M.~Kyropoulou, B.~Lucier,
  R.~{Paes Leme}, and {\'E}.~Tardos.
\newblock {Bounding the inefficiency of outcomes in generalized second price
  auctions}.
\newblock {\em Journal of Economic Theory}, 156:343--388, 2015.
\newblock Computer Science and Economic Theory.

\bibitem{Cho21}
T.~Chotibut, F.~Falniowski, M.~Misiurewicz, and G.~Piliouras.
\newblock The route to chaos in routing games: When is price of anarchy too
  optimistic?
\newblock In H.~Larochelle, M.~Ranzato, R.~Hadsell, M.~F. Balcan, and H.~Lin,
  editors, {\em Advances in Neural Information Processing Systems}, volume~33,
  pages 766--777, Red Hook, NY, USA, 2020. Curran Associates, Inc.

\bibitem{chotibut20family}
T.~Chotibut, F.~Falniowski, M.~Misiurewicz, and G.~Piliouras.
\newblock Family of chaotic maps from game theory.
\newblock {\em Dynamical Systems}, 36(1):48--63, 2021.

\bibitem{conitzer2006failures}
V.~Conitzer and T.~Sandholm.
\newblock Failures of the vcg mechanism in combinatorial auctions and
  exchanges.
\newblock In {\em Proceedings of the Fifth International Joint Conference on
  Autonomous Agents and Multiagent Systems}, AAMAS '06, page 521–528, New
  York, NY, USA, 2006. Association for Computing Machinery.

\bibitem{Li75}
T.-Y. Li and J.~A. Yorke.
\newblock {Period Three Implies Chaos}.
\newblock {\em The American Mathematical Monthly}, 82(10):985--992, 1975.

\bibitem{Bre12}
Brendan Lucier, Renato Paes~Leme, and Eva Tardos.
\newblock On revenue in the generalized second price auction.
\newblock In {\em Proceedings of the 21st International Conference on World
  Wide Web}, WWW '12, page 361–370, New York, NY, USA, 2012. Association for
  Computing Machinery.

\bibitem{Mon20}
B.~Monnot, H.Q. Ze, C.S.M. Koh, S.~Leonardos, and G.~Piliouras.
\newblock {Ethereum's Transaction Fee Market Reform of EIP 1559}.
\newblock In {\em Proceedings of the WINE 2020 Workshop on Game Theory in
  Blockchain}, pages 7--7, 2020.

\bibitem{Nis07}
Noam Nisan, Tim Roughgarden, Eva Tardos, and Vijay~V. Vazirani.
\newblock {\em Algorithmic Game Theory}.
\newblock Cambridge University Press, USA, 2007.

\bibitem{Pal17}
Gerasimos Palaiopanos, Ioannis Panageas, and Georgios Piliouras.
\newblock Multiplicative weights update with constant step-size in congestion
  games: Convergence, limit cycles and chaos.
\newblock In {\em Proceedings of the 31st International Conference on Neural
  Information Processing Systems}, NIPS'17, page 5874–5884, Red Hook, NY,
  USA, 2017. Curran Associates Inc.

\bibitem{Rob20e}
{Robust Incentives Group, Ethereum Foundation}.
\newblock {EIP 1559: A transaction fee market proposal}, 2020.
\newblock
  \href{https://nbviewer.jupyter.org/github/ethereum/rig/blob/master/eip1559/eip1559.ipynb}{}
  [Online; accessed 11-February-2021].

\bibitem{roughgarden2020transaction}
T.~Roughgarden.
\newblock Transaction fee mechanism design for the ethereum blockchain: An
  economic analysis of eip-1559.
\newblock In {\em Proceedings of the 21st ACM Conference on Economics and
  Computation}, EC '21, page (to appear), New York, NY, USA, 2021. Association
  for Computing Machinery.
\newblock arXiv:2012.00854.

\bibitem{Fil20}
Filecoin Staff.
\newblock Filecoin features: Gas fees, 2020.
\newblock
  \href{https://filecoin.io/blog/filecoin-features-gas-fees/}{filecoin.io}
  [Online; accessed: 11-February-2021].

\bibitem{Var07}
Hal~R. Varian.
\newblock Position auctions.
\newblock {\em International Journal of Industrial Organization},
  25(6):1163--1178, 2007.

\bibitem{wood2014ethereum}
Gavin Wood et~al.
\newblock Ethereum: A secure decentralised generalised transaction ledger.
\newblock {\em Ethereum project yellow paper}, 151(2014):1--32, 2014.

\end{thebibliography}
